\documentclass[acmsmall]{acmart}
\AtBeginDocument{%
  }

\setcopyright{cc}
\setcctype{by}

\acmJournal{PACMMOD}
\acmYear{2026} \acmVolume{4} 
\acmNumber{3 (SIGMOD)} \acmArticle{128}
\acmMonth{6} \acmDOI{10.1145/3802005}

\copyrightyear{2026}

\usepackage[ruled,vlined,linesnumbered]{algorithm2e} 
\usepackage{soul}
\usepackage{amsmath}
\usepackage{amsthm}
\usepackage{cleveref}
\usepackage{subcaption}
\usepackage{multirow}
\usepackage{enumitem}

\setitemize{leftmargin=10pt}
\newcommand{\cmark}{$\checkmark$}%
\newcommand{\xmark}{$\times$}%

\RestyleAlgo{ruled} 

\newtheorem{definition}{Definition}[section]
\newtheorem{theorem}{Theorem}[section]
\newtheorem{lemma}{Lemma}
\newtheorem{example}{Example}[section]
\newtheorem{rulethm}{Rule}
\newtheorem*{rulethm*}{Rule}

\begin{document}

\title{Accelerating Maximum Common Subgraph Computation by Exploiting Symmetries}


\author{Buddhi Kothalawala}
\email{buddhi.kothalawala@anu.edu.au}
\orcid{0000-0003-3023-2905}
\affiliation{%
  \department{School of Computing}
  \institution{The Australian National University}
  \city{Canberra}
  \state{The Australian Capital Territory}
  \country{Australia}
}

\author{Henning Koehler}
\email{h.koehler@massey.ac.nz}
\orcid{0000-0002-4688-920X}
\affiliation{%
  \department{School of Mathematical and Computational Sciences}
  \institution{Massey University}
  \city{Palmerston North}
  \country{New Zealand}
}

\author{Muhammad Farhan}
\email{muhammad.farhan@anu.edu.au}
\orcid{0000-0002-1239-2107}
\affiliation{
  \department{School of Computing}
  \institution{The Australian National University}
  \city{Canberra}
  \state{The Australian Capital Territory}
  \country{Australia}
}


\begin{abstract}
The Maximum Common Subgraph (MCS) problem plays a key role in many applications, including cheminformatics, bioinformatics, and pattern recognition, where it is used to identify the largest shared substructure between two graphs. Although symmetry exploitation is a powerful means of reducing search space in combinatorial optimization, its potential in MCS algorithms has remained largely underexplored due to the challenges of detecting and integrating symmetries effectively. Existing approaches, such as RRSplit, partially address symmetry through vertex-equivalence reasoning on the variable graph, but symmetries in the value graph remain unexploited. In this work, we introduce a complete dual-symmetry breaking framework that simultaneously handles symmetries in both variable and value graphs. Our method identifies and exploits modular symmetries based on local neighborhood structures, allowing the algorithm to prune isomorphic subtrees during search while rigorously preserving optimality. Extensive experiments on standard MCS benchmarks show that our approach substantially outperforms the state-of-the-art RRSplit algorithm, solving more instances with significant reductions in both computation time and search space. These results highlight the practical effectiveness of comprehensive symmetry-aware pruning for accelerating exact MCS computation.

\end{abstract}



\received{October 2025}
\received[revised]{January 2026}
\received[accepted]{February 2026}

\maketitle

\section{Introduction}

\begin{figure}[t]
\centering
\includegraphics{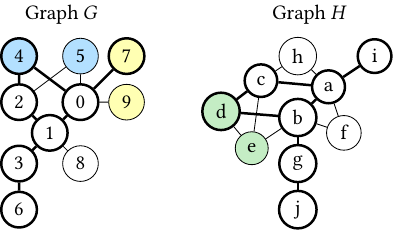}
\caption{Two example graphs $G$ and $H$, with symmetric vertex sets $\{4, 5\}$, $\{7, 9\}$, and $\{d, e\}$. The pairs $(4, 5)$ and $(7, 9)$ exhibit negative symmetry, while $(d, e)$ exhibit positive symmetry. The maximum common induced subgraph (MCIS) is shown in bold, with the vertex mapping $M = \{(0, a), (1, b), (2, d), (3, g), (4, c), (6, j), (7, i)\}$.}
\label{fig:graphs}
\Description{MCIS example graphs}
\vspace{-0.2cm}
\end{figure}

Graphs are a powerful tool for representing structured data and are widely used in diverse application domains. For instance, computer networks are modeled with devices as vertices and communication links as edges; chemical compounds are represented as graphs where atoms and bonds correspond to vertices and edges \citep{ehrlich2011chemanalisis, kawabata2014chemapplication}; and social networks model individuals and their interactions in a similar manner \citep{lahiri2008mining-network}. Identifying structural similarities in graphs is critical for tasks such as detecting anomalies in network traffic, classifying chemical compounds, and analyzing communities in social networks \citep{vijayalakshmi2011performance-network}.

A well-established measure of similarity between two graphs is the \emph{Maximum Common Induced Subgraph} (MCIS). The MCIS problem asks for the largest induced subgraphs of two input graphs that are isomorphic, meaning there exists a one-to-one correspondence between their vertices that preserves adjacency. \Cref{fig:graphs} illustrates two example graphs, \( G \) and \( H \), with their maximum common induced subgraph shown in bold. Since MCIS is NP-hard \citep{akutsu2012complexity}, finding exact solutions is computationally challenging, particularly for large and dense graphs.

To tackle the MCIS problem, a variety of algorithmic strategies have been explored, including constraint programming approaches \citep{vismara2008cpmethod, ndiaye2011cpmethod}, clique-based formulations \citep{levi1973note, koch2001enumerating, balas1986finding}, and backtracking-based search algorithms \citep{mcgregor1982backtracksa, mccreesh2017partitioning}. Among the most effective exact methods is the McSplit algorithm, a branch-and-bound (BnB) approach that systematically explores the space of partial vertex mappings using a depth-first search strategy \citep{mccreesh2017partitioning}. McSplit partitions candidate vertex pairs into \emph{bidomains}, each representing subsets of vertices from \( G \) and \( H \) that remain feasible for matching under the current mapping. It then applies a bounding rule to estimate the largest possible extension of the current mapping, allowing early pruning of branches that cannot lead to an optimal solution. Building on McSplit, the recently proposed RRSplit algorithm \citep{yu2025rrsplit} improves efficiency by introducing vertex-equivalence-based pruning and a refined upper bound that better reflects local structural redundancy. By grouping structurally equivalent vertices into equivalence classes, RRSplit avoids redundant exploration of symmetric subtrees and achieves notable performance gains. However, its treatment of symmetry remains limited: the equivalence reasoning applies only within one input graph and does not account for the symmetries that arise from both input graphs. 
For example, in \cref{fig:graphs}, the vertices in graph $G$ that share identical neighborhoods are $\{4, 5\}$ and $\{7, 9\}$, while in $H$, the corresponding pair is $\{d, e\}$.
For instance, in a common subgraph containing at least one of vertices $4$ and $5$, swapping them produces a symmetric solution.
RRSplit detects and avoids revisiting such redundant cases.
However, swapping vertices $\{d, e\}$ in $H$ also produces symmetric solutions, provided at least one of them is contained in the common subgraph.
Such redundant cases caused by equivalent vertices in $H$ are missed by RRSplit, whereas our proposed method accounts for symmetries in both input graphs, $G$ and $H$.
While this seems straight-forward conceptually, these symmetries can interact, making correctness and efficient detection/avoidance non-trivial issues.

In the broader context of constraint programming, the two input graphs in an MCIS instance are often referred to as the \emph{variable graph} and the \emph{value graph}. Vertices in the variable graph are treated as variables to be assigned, while vertices in the value graph, along with a special symbol representing unmatched vertices, serve as possible values. The objective of an MCIS algorithm is to assign values to variables so that the induced subgraphs of the variable and value graphs are isomorphic and no other mapping includes more vertices. This perspective highlights the asymmetric yet interdependent roles of the two graphs and motivates the need for symmetry breaking strategies that operate coherently across both.

To address these limitations, we propose a novel algorithm called \emph{SymSplit}, which systematically eliminates redundant exploration by breaking symmetries arising from similar neighborhoods in both the variable and value graphs. At the core of our approach lies a dual-symmetry breaking framework that preserves the exactness of the algorithm while substantially improving efficiency. SymSplit identifies symmetric structures in the input graphs and prunes equivalent search branches that correspond to the same subgraph mappings. By avoiding these unnecessary explorations, the proposed method significantly reduces computation time without compromising completeness or correctness.

\medskip
\noindent\textbf{Contributions.~}The main contributions of this work are summarized as follows:
\begin{itemize}[itemsep=5pt]
    \item We present a principled framework for detecting and eliminating symmetric branches during the search by identifying structurally equivalent vertices with identical neighborhoods in both input graphs. The framework unifies two complementary techniques: variable symmetry breaking and value symmetry breaking that together provide more complete symmetry elimination for the MCIS problem.
    
    \item We design an efficient and lightweight method for detecting relevant graph symmetries that minimizes computational overhead. Our implementation ensures that the symmetry detection process scales effectively to large graph instances, making the approach practical for real-world applications.
    
    \item We formally establish the correctness and completeness of the proposed framework, proving that variable and value symmetry breaking can be applied simultaneously without interference.
    Our proofs extend to directed graphs and graphs with self-loops. We also analyze the computational complexity of symmetry identification and computation.
    
    \item We perform extensive experiments on standard MCIS benchmark datasets to assess the effectiveness of our approach. The results show that \emph{SymSplit} consistently solves more problem instances in less time compared to state-of-the-art algorithms, demonstrating the practical benefits of comprehensive symmetry breaking in exact MCIS computation.
\end{itemize}

\section{State-of-the-art Solutions}\label{sec:prelim}
Let \( G = (V(G), E(G)) \) be a simple and unweighted graph, where \( V(G) \) and \( E(G) \) denote the vertex and edge sets, respectively. For any subset \( V' \subseteq V(G) \), the \emph{induced subgraph} \( G[V'] \) is defined as the graph with vertex set \( V' \) and edge set \( \{(v_i, v_j) \in E(G) \mid v_i, v_j \in V'\} \). Two graphs \( G \) and \( H \) are said to be \emph{isomorphic}, denoted \( G \cong H \), if there exists a bijection \( \varphi : V(G) \to V(H) \) such that for all \( v_i, v_j \in V(G) \), 
\((v_i, v_j) \in E(G)\) if and only if \((\varphi(v_i), \varphi(v_j)) \in E(H)\).

\begin{definition}[Maximum Common Induced Subgraph]
The \emph{maximum common induced subgraph} (MCIS) problem for two graphs \( G \) and \( H \) is to find induced subgraphs \( G[V'] \subseteq G \) and \( H[U'] \subseteq H \) such that \( G[V'] \cong H[U'] \) and \( |V'| = |U'| \) is maximized.
\end{definition}

The MCIS problem is NP-hard and serves as a fundamental model for identifying structural similarity between graphs. Most exact approaches employ a \emph{branch-and-bound} strategy that incrementally builds vertex correspondences while pruning infeasible or suboptimal branches using upper bounds. Two representative algorithms following this framework are \emph{McSplit} and its refinement \emph{RRSplit}, which are discussed next.

\subsection{McSplit Algorithm}
The \emph{McSplit algorithm}~\cite{mccreesh2017partitioning} applies the branch-and-bound paradigm to the MCIS problem. Given two graphs $G$ and $H$, it searches for a \emph{mapping} of maximum cardinality that represents the largest common induced subgraph between them.

A mapping between $G$ and $H$ is a set of vertex pairs
\[
M = \{(v_1, u_1), (v_2, u_2), \dots (v_i, u_i) \dots, (v_k, u_k)\},
\]
where each vertex in $v_i \in V(G)$ and $u_i \in V(H)$ appears in at most one pair. Each pair $(v_i, u_i) \in M$ is called a \emph{mapped vertex pair}. A mapping of maximum cardinality corresponds directly to a maximum common induced subgraph of $G$ and $H$. McSplit performs a depth-first search (DFS) beginning with the empty mapping and incrementally extends it by adding one compatible vertex pair at a time, continuing until no further extensions are possible.

At the core of McSplit lies a \emph{bidomain partitioning} mechanism that organizes unmatched vertices according to their structural compatibility. This compatibility is determined using an \emph{M-labelling} scheme that records the adjacency of each unmatched vertex to the vertices already in the current mapping. Specifically, given a partial mapping 
$M = \{(v_1, u_1), \dots, (v_k, u_k)\}$ 
with corresponding vertex sets $V' = \{v_1, \dots, v_k\}$ and $U' = \{u_1, \dots, u_k\}$, the labelling function
\[
\ell : V(G) \cup V(H) \to \{0,1\}^k
\]
assigns a binary vector to each unmatched vertex $w \in (V(G) \setminus V') \cup (V(H) \setminus U')$, where the $i$th bit of $\ell(w)$ equals $1$ if $w$ is adjacent to $v_i$ in $G$ or to $u_i$ in $H$, and $0$ otherwise. Two vertices receive identical labels if they have the same adjacency relationships with respect to the current mapping.

Using these labels, McSplit partitions the unmatched vertices into a collection of bidomains 
\[
P_M = \{\langle V_1, U_1 \rangle, \dots, \langle V_t, U_t \rangle\},
\]
such that for every $v \in V_i$ and $u \in U_j$, $i = j$ if and only if $\ell(v) = \ell(u)$. Vertices in the same bidomain pair share similar structural roles, while different bidomains correspond to distinct adjacency patterns.

The algorithm begins with a single bidomain $P_M = \{\langle V(G), V(H) \rangle\}$ and recursively refines it through branching. At each step, a bidomain $\langle V_l, U_l \rangle$ is selected, and a vertex $v \in V_l$ is chosen to be matched against each vertex $u \in U_l$. Each resulting mapping $(M \cup \{(v,u)\})$ defines a new search node. After adding $(v,u)$, McSplit relabels the remaining unmatched vertices based on their connectivity to $v$ ($N_G(v)$ neighborhood of $v$) and $u$ ($N_H(u)$ neighborhood of $u$). This refinement splits the current bidomain $\langle V_l, U_l \rangle$ into two smaller bidomains:
\[
\langle V_l \cap N_G(v),\, U_l \cap N_H(u) \rangle
\quad \text{and} \quad
\langle V_l \setminus N_G(v ),\, U_l \setminus N_H(u) \rangle,
\]
representing vertices adjacent and non-adjacent to the selected pair, respectively. The process continues recursively. Once $v$ has been matched with all vertices in $U_l$, it is removed from consideration by treating it as paired with a special symbol $\bot$. This algorithm follows a tree search, which we can formally define the search tree as follows.

\vspace{0.15cm}
\noindent\textbf{Search Tree.~}A \emph{search tree} $T$ for graphs $G$ and $H$ is a rooted tree where each node represents a partial mapping $M$ with bidomains $P_M$. The root represents the empty mapping $M_0 = \emptyset$ with $P_{M_0} = \{\langle V(G), V(H) \rangle\}$. Each child extends its parent by adding $(v_i, u_i)$ to map $v_i \in V(G)$ to $u_i \in V(H)$, or $(v_i, \bot)$ to leave $v_i$ unmapped. Bidomains are refined based on adjacency and label constraints, and all mappings induce isomorphic subgraphs. Nodes with no valid extensions are leaves. A \emph{subtree} $T_M$ rooted at a node with mapping $M$ consists of all descendant nodes reachable by extending $M$. A \emph{branch} is a path from the root to a leaf, representing a complete sequence of mapping decisions.

Throughout the search, McSplit maintains the best mapping found so far, referred to as the \emph{incumbent}. To avoid exploring unpromising branches, it applies a bounding rule that estimates the largest possible mapping size achievable from the current state. The upper bound is computed as
\begin{equation}\label{Eqn: McSplitBound}
\mathrm{UB}(M, P_M) = |M| + \sum_{\langle V_l, U_l \rangle \in P_M} \min(|V_l|, |U_l|).
\end{equation}
If $\mathrm{UB}(M, P_M) \le |M^*|$, where $M^*$ is the incumbent, the current branch is pruned; otherwise, the search proceeds recursively. This bound guarantees correctness while significantly reducing unnecessary exploration.

\begin{example}
\Cref{fig:search-tree} illustrates the McSplit search process partially. The root node begins with all vertices of $G$ and $H$ in a single bidomain and an empty mapping. Suppose the algorithm first selects the pair $(0,a)$, where vertex $0$ in $G$ is adjacent to $\{1,4,5,7,9\}$ and vertex $a$ in $H$ is adjacent to $\{b,c,f,h,i\}$. These adjacent vertices form a new bidomain, while the remaining vertices form another, resulting in a refined partition $P_{M_1}$ for the mapping $M_1 = \{(0,a)\}$. Next, mapping $(3,d)$ is chosen, and $P_{M_1}$ is further refined based on adjacency to $3$ in $G$ and $d$ in $H$. In particular, the bidomain $\langle \{1,4,5,7,9\}, \{b,c,f,h,i\} \rangle$ is split into $\langle \{1\}, \{b,c\} \rangle$ for vertices adjacent to both $3$ and $d$, and $\langle \{4,5,7,9\}, \{f,h,i\} \rangle$ for those not connected. This refinement proceeds recursively until no further extensions are possible, at which point the search branch is terminated.
\end{example}

\subsection{RRSplit Algorithm}
The RRSplit algorithm~\cite{yu2025rrsplit} builds upon McSplit by addressing two key limitations of its search process: redundant exploration caused by structurally equivalent vertices and an overly loose upper bound that weakens pruning efficiency. To overcome these issues, RRSplit introduces three interdependent mechanisms that together create a redundancy-reduced branch-and-bound framework.

RRSplit operates on the same principle as McSplit, maintaining a partial mapping $M$ where each vertex in $G$ and $H$ appears in at most one pair. The key difference lies in how RRSplit exploits vertex equivalence to avoid redundant exploration and computes tighter upper bounds for pruning. Two vertices are said to be \emph{structurally equivalent} if they share the same neighborhood: 
\[
v_1 \sim v_2 \quad \text{iff} \quad N_G(v_1) = N_G(v_2),
\]
Such vertices play identical structural roles in any induced subgraph isomorphism. RRSplit partitions the vertices of $G$ into equivalence classes, where each class $\Psi(v)$ contains all vertices structurally equivalent to $v$. During the search, when a vertex $v$ in $G$ is selected for matching, only one configuration of its equivalence class is explored, while the remaining symmetric configurations are pruned. For example, if $v_1, v_2 \in \Psi_G(v)$ and $u_1, u_2 \in V(H)$, the mappings
\[
M_1 = \{(v_1, u_1), (v_2, u_2)\} \quad \text{and} \quad M_2 = \{(v_1, u_2), (v_2, u_1)\},
\]
yield search subtrees with isomorphic subgraphs, and RRSplit explores only one of them. The same logic applies when a vertex is excluded from the mapping: if one vertex of an equivalence class is removed, other equivalent vertices need not be considered separately. This vertex-equivalence-based reduction removes redundant subtrees that McSplit would otherwise explore, particularly in graphs with repeated local structures or regular patterns.

RRSplit further strengthens the McSplit upper bound (\ref{Eqn: McSplitBound}) by incorporating vertex equivalence information to ignore matches that would be symmetry pruned.
However, the practical gains from this bound refinement are modest, as discussed in Section~\ref{sec:ablation}.

\begin{example}
In the RRSplit algorithm, vertices $\{4, 5\} \subseteq V(G)$ and $\{7, 9\} \subseteq V(G)$ in Figure~\ref{fig:graphs} are identified as structurally equivalent. Consider the two mappings $M_8 = \{(0, b), (2, c), (7, f), (9, \bot)\}$ and $M_{10} = \{(0, b), (2, c), (7, \bot), (9, f)\}$ shown under variable symmetry in Figure~\ref{fig:search-tree}. When RRSplit detects that $v_1, v_2 \in \Psi_G(v_1)$ belong to equivalent vertex sets and $u_1, u_2 \in V(H)$, it explores only one representative mapping $M_1 = \{(v_1, u_1), (v_2, u_2)\}$ and prunes the symmetric counterpart $M_2 = \{(v_1, u_2), (v_2, u_1)\}$. Here, $v_1 = 7$, $v_2 = 9$, $u_1 = f$, and $u_2 = \bot$. Consequently, RRSplit avoids exploring the branch corresponding to $M_{10}$ because it represents a symmetric configuration of $M_8$. 
\end{example}

\subsection{Discussion}
Although McSplit and RRSplit show significant progress in solving the MCIS problem, both remain limited by incomplete handling of structural symmetry. McSplit does not recognize equivalent structural configurations within or across the input graphs at all, often exploring multiple branches that lead to equivalent subgraphs.

RRSplit improves upon McSplit by introducing vertex-equiva\-lence-based pruning. However, its symmetry handling is inherently partial: it focuses only on equivalences within the variable graph $G$ but ignores equivalences within $H$.

Our proposed dual-symmetry breaking framework directly addresses this gap by jointly eliminating symmetries in both graphs $G$ and $H$ while ensuring to return the exact MCIS, thus enabling a more efficient and fully symmetry-aware search process.

It is important to note that symmetries arising from equivalences in \( H \) are conceptually similar to those in \( G \); however, eliminating them presents significant challenges. These difficulties stem from the asymmetric roles that \( G \) and \( H \) assume within the branch-and-bound framework, as well as from the need to address potential conflicts when applying variable and value symmetry-breaking techniques simultaneously.

\begin{figure}[t]
\centering
\resizebox{\textwidth}{!}{\includegraphics{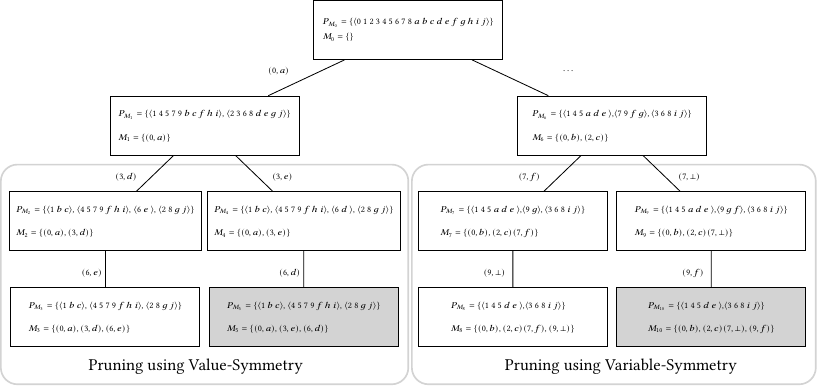}}
\caption{An illustration of a partial MCIS search tree demonstrating the proposed symmetry breaking rule. Grey nodes represent subproblems that are isomorphic to their left siblings and are pruned using either value-symmetry or variable-symmetry. The left box shows pruning based on H-symmetry, while the right box shows pruning based on variable-symmetry.}
\label{fig:search-tree}
\Description{Search Tree}
\end{figure}

\section{Dual-Symmetry Breaking Framework}\label{sec:approach}

Symmetry breaking has been widely studied in combinatorial search problems related to MCIS, particularly within the constraint programming framework~\cite{fahle2001symmetry, zampelli2007symmetry}. There, the two input graphs play distinct roles: the graph \( G \) is treated as the \emph{variable graph}, whose vertices are successively selected for mapping, while \( H \) serves as the \emph{value graph}, whose vertices (and the special symbol~\(\bot\)) represent possible assignments to these variables. In this setting, symmetry breaking aims to prevent redundant exploration of equivalent search branches that yield isomorphic or structurally identical solutions.

An effective symmetry-breaking mechanism must address symmetries arising from both graphs and ensure that they can be handled concurrently. Specifically, three components are essential:  
(1) \emph{variable symmetry breaking}, which eliminates redundant mappings originating from symmetric vertices in the variable graph \( G \);  
(2) \emph{value symmetry breaking}, which removes equivalent mappings induced by symmetries in the value graph \( H \); and  
(3) \emph{compatibility}, which guarantees that at least one symmetrical branch is preserved by allowing variable and value symmetry breakings to be applied together.
A framework satisfying all three criteria constitutes a \emph{complete symmetry-breaking} solution.

Note that compatibility of variable and value symmetry breaking rules is not automatic. E.g. for two mappings $M_1,M_2$ which are related through both value and variable symmetry, one rule may eliminate $M_1$ and assume that $M_2$ is preserved, while the other eliminates $M_2$ assuming $M_1$ is kept.

Building on this principle, we introduce a \emph{dual-symmetry breaking framework} for the MCIS problem that simultaneously applies symmetry-breaking constraints to both \( G \) and \( H \). The proposed framework integrates two complementary mechanisms:
\begin{itemize}
    \item \textbf{Variable Symmetry Breaking:} Exploits structural equivalences within the variable graph \( G \) to avoid exploring isomorphic subtrees caused by symmetric vertex orderings.
    \item \textbf{Value Symmetry Breaking:} Leverages symmetries in the value graph \( H \) to prevent redundant mappings arising from interchangeable vertices or identical adjacency patterns.
\end{itemize}
Any proposed MCIS solution must satisfy certain criteria to ensure it solves the MCIS problem correctly. We define the properties that any MCIS algorithm should adhere to as follows:
\begin{definition}[Properties of MCIS Algorithms]
An algorithm for the MCIS problem must satisfy the following properties:
\begin{itemize}
    \item \textbf{Correctness:} Only mappings between isomorphic induced subgraphs of \( G \) and \( H \) are considered valid.
    \item \textbf{Completeness:} At least one mapping of maximal cardinality is guaranteed to be found.
\end{itemize}
\end{definition}

The rest of this section is structured as follows. First, we discuss our symmetry detection technique, which we call \emph{modular symmetry}. We then individually discuss variable and value symmetry. Following this, we explain how these symmetries can be applied simultaneously to the MCIS problem without violating completeness. Next, we present the proposed algorithm and demonstrate how the method extends to directed graphs and graphs with self-loops. Finally, we analyze the computational complexity of the proposed symmetry breaking technique.

\subsection{Modular Symmetry}
To identify structural symmetries in the input graphs $G$ and $H$, we focus on detecting symmetric vertices by analyzing the similarity of their neighborhoods. The core idea is that vertices with identical neighborhoods play analogous structural roles in the graph and exhibit similar behavior under the McSplit labeling function. By grouping such vertices into symmetric classes, we can uncover meaningful structural patterns and reduce computational redundancy during the search process.

First, we define two types of neighborhoods that form the basis of our symmetry.

\begin{definition}[Neighborhood]
\label{def:neighborhood}
For any vertex $v \in V(G)$, the \emph{negative} and \emph{positive neighborhoods} of $v$ can be defined as follows:
\begin{itemize}
    \item \emph{Negative neighborhood}: $N^-(v) = \{ u \in V \mid (u, v) \in E \}$
    \item \emph{Positive neighborhood}: $N^+(v) = N^-(v) \cup \{v\}$
\end{itemize}
\end{definition}

\begin{definition}[Modular Symmetry]
\label{def:modular-sym}
For a graph \( G \), two vertices \( u, v \in V(G) \) are said to be \emph{modular symmetric}, denoted $u \equiv v$, if they share the same negative neighborhood or the same positive neighborhood. Specifically, $u$ and $v$ are said to be negatively symmetric if $N^-(u) = N^-(v)$, and positively symmetric if $N^+(u) = N^+(v)$.
\end{definition}

Modular symmetry captures both adjacency-based and self-inclusive equivalences. It generalizes the notion of structural equivalence used in RRSplit~\cite{yu2025rrsplit}. By incorporating both neighborhood types, our formulation exposes additional symmetric patterns that would otherwise remain undetected.

An important property of modular symmetry is that positive and negative symmetries are mutually exclusive.

\begin{theorem}[Mutual Exclusiveness]
\label{thm:mutual-exclusiveness}
A vertex $v$ cannot participate in both negative and positive symmetry relationships.
\end{theorem}
\begin{proof}
Suppose \( v, u, w \in V(G) \), and assume \( v \) and \( u \) are positively symmetric, i.e., \( N^+(v) = N^+(u) \), which implies in particular that \( (v, u) \in E(G) \). Now suppose \( v \) and \( w \) are negatively symmetric, i.e., \( N^-(v) = N^-(w) \), and by definition of negative symmetry, \( (v, w) \notin E(G) \).
Since \( v \) and \( w \) have identical negative neighborhoods, and \( v \) is adjacent to \( u \), it follows that \( w \) must also be adjacent to \( u \). However, as \( (v, w) \notin E(G) \), it follows that \( v \) is not adjacent to \( w \), but \( u \) is adjacent to \( w \). This contradicts
\( N^+(v) = N^+(u) \).
\end{proof}

Modular symmetry thus induces equivalence classes over the vertex set $V(G)$. Each class contains vertices that are pairwise symmetric under one of the two criteria. A positive symmetry class always forms a \emph{clique}, meaning every pair of vertices in the class is connected by an edge, while each negative symmetry class forms an \emph{anti-clique}, also known as an independent set (i.e., a vertex set without any edge between them).

\begin{example}
\label{ex:symmetry-example}
Consider graphs $G$ and $H$ in \cref{fig:graphs}. In graph $G$, the vertex pairs $\{4, 5\}$ are structurally symmetric as they share the same set of neighbors: $\{0, 2\}$. Hence, they belong to the same symmetry class. Similarly, vertices $7$ and $9$ in graph $G$, and vertices $d$ and $e$ in graph $H$, are symmetric based on their neighborhood structures. Especially, $\{d, e\}$ pair is symmetric under positive modular symmetry.
\end{example}

We note that modular symmetry can equivalently be viewed through an automorphism perspective: two vertices \( u, v \) are modular symmetric if the permutation swapping \( u \) and \( v \) is an automorphism of \( G \). Under this interpretation, transitivity follows immediately, although constructing equivalence classes directly is less straightforward.

As a side note for readers familiar with modules and modular decompositions, our notion of symmetry is indeed related. The non-trivial equivalence classes induced by modular symmetry are precisely the non-trivial strong modules that are cliques or anti-cliques.

\subsection{Variable Symmetry Breaking}
We now describe symmetry breaking applied to the variable graph \( G \). We call two mappings \emph{symmetrical} iff they can be converted into each other by swapping out modular symmetrical vertices.

\begin{rulethm}[Variable Symmetry Breaking]\label{rule:var-sym}
    Consider a search branch $b$ containing $(v_1, u_2), (v_2, u_1)$ where $v_1,v_2\in G$ and $u_1,u_2\in H\cup\{\bot\}$ with $v_1 \equiv v_2$ and $v_1$ appearing before $v_2$ in $b$. If $u_1 < u_2$ w.r.t. some arbitrary but fixed ordering, we prune $b$.
\end{rulethm}

\noindent\textbf{Justification.} Since $v_1$ and $v_2$ have identical neighborhoods, swapping their assignments in any mapping preserves the isomorphism property while creating a more consistently ordered mapping.

We say that edges $(v_1,u_2),(v_2,u_1)$ are in \emph{$G$-conflict} on a branch mapping $v_1$ prior to $v_2$ iff $v_1\equiv v_2$ and $u_1 < u_2$. Observe that $G$-symmetry prunes exactly the branches containing $G$-conflicts.
\begin{theorem}\label{th:var-sym}
    Variable symmetry breaking does not violate the correctness or completeness of the MCIS solution.
\end{theorem}

\begin{proof}
Let the conflicting edges be labeled $(v_1,u_2),(v_2,u_1)$ where possibly $u_2=\bot$, for a branch $b$ with mapping $M$.
By performing a swap of $(v_1,u_2),(v_2,u_1)$ to $(v_1,u_1),(v_2,u_2)$ we can create a symmetrical mapping $M'$ with corresponding hypothetical search branch $b'$.
We call $b'$ hypothetical since it may have been pruned as well, due to a different $G$-conflict.
However, since search trees branch on all possible mappings of $v$ at the same node, $b'$ must have been considered before $b$, assuming values are consider in the same arbitrary but fixed order used in Rule~\ref{rule:var-sym} (more on this assumption later).
Thus repeated swapping of $G$-conflicting edges will lead to earlier and earlier hypothetical branches, eventually resulting in an actual branch.
This shows completeness, correctness is trivial.
\end{proof}

We note that our approach to variable symmetry breaking is similar to the \emph{vertex-equivalence-based reduction} used by RRSplit.
One of the key differences is that RRSplit does not make the assumption about a fixed processing order of $H$-vertices used in the proof of Theorem~\ref{th:var-sym}, and instead ensures that $b'$ would have been considered earlier by tracking actual processing order.

However, this assumption is not necessary to show that the sequence of hypothetical branches eventually terminates.
Instead, consider the following (quasi-)order between branches:

\begin{definition}[Value-Lexicographical Order]\label{def:vl-order}
Let $<_H$ be an arbitrary but fixed order on $V(H)\cup\{\bot\}$.
The \emph{value-lexicographical} order $\prec$ compares branches based on the $<_H$-based lexicographical order of their sequences of $H$-vertices or $\bot$ values.
\end{definition}

The swaps performed in the proof of Theorem~\ref{th:var-sym} always result in a branch which is strictly smaller w.r.t. value-lexicographical order, thus ensuring termination.
Apart from being of theoretical interest, this simplifies implementation by removing the need to either impose a fixed processing order of $H$-vertices or tracking it.

\begin{example}[Variable Symmetry Breaking]
Consider the right side of the tree shown in Figure~\ref{fig:search-tree}, the mappings \((7, f)\) with \(9\) unmatched, and \((9, f)\) with \(7\) unmatched, leading to isomorphic nodes \(M_8\) and \(M_{10}\). Here, the we identify that vertex \(9\) is symmetric to an already visited vertex \(7\), which has been mapped earlier. As a result, it avoids exploring the second branch using $G$-symmetry pruning.
\end{example}

\subsection{Value Symmetry Breaking}
Similarly, vertices in \( H \) may have identical structural properties, leading to redundant choices when selecting mapping targets. We can address this in a way analogous to variable symmetry.

\begin{rulethm*}[Value Symmetry Breaking -- Basic]\label{rule:val-sym}
    Consider a search branch $b$ containing $(v_1, u_2), (v_2, u_1)$ where $v_1,v_2\in G$ and $u_1,u_2\in H$ with $u_1 \equiv u_2$ and $v_1$ appearing before $v_2$ in $b$. If $u_1 < u_2$ w.r.t. some arbitrary but fixed ordering, we prune $b$.
\end{rulethm*}

We shall refer to such mapping pairs as being \emph{$H$-conflicting}. Completeness of this rule can now be shown just as for variable symmetry breaking.
However, we can prune branches even earlier, and at the same time reduce the overhead for pruning checks, by considering all possible extensions of a search branch, including implicit mappings $(\bot,u)$ that leave a vertex $u\in V(H)$ unmatched.
We observe that when mapping a branching vertex $v$ to a variable vertex $u_2$ for which a modular symmetrical vertex $u_1\equiv u_2$ has not been mapped yet, all possible extensions of this branch will lead to $H$-conflicts and could thus be pruned.
Here we permit $v_1,v_2\in G\cup\{\bot\}$ and consider $\bot$ to appear last in any branch, for the purpose of assessing $H$-conflicts.
In such a case it makes sense to prune immediately rather than wait until the $H$-conflicts finally manifest, which leads to the following improved pruning rule.

\begin{rulethm}[Value Symmetry Breaking -- Enhanced]\label{Rule: Value Symmetry Breaking}
    When there exists $u_1, u_2, \dots, u_n \in V(H)$ such that $u_1 \equiv u_2 \equiv \dots \equiv u_n$ with an arbitrary but fixed order $u_1 < u_2 < \dots < u_n$, then we always pick $u_1$ for matching and ignore the rest.
\end{rulethm}

By restricting branching to the smallest representative of each symmetry class, the search avoids exploring multiple equivalent mappings differing only by the permutation of symmetric vertices in \( H \). This ensures that redundant permutations are pruned preemptively.

\begin{theorem}
    \label{thm:val-symmetry}
     Value symmetry breaking does not violate the correctness or completeness of the MCIS solution.
\end{theorem}

\begin{proof}
Let $u_1, u_2 \in V(H)$ with $u_1 \equiv u_2$, $u_1 < u_2$, and $v_1, v_2 \in V(G)$. Suppose the conflicting edges are labeled $(v_1,u_2)$ and $(v_2,u_1)$ for a branch $b$ with mapping $M$. By swapping $(v_1,u_2),(v_2,u_1)$ to $(v_1,u_1),(v_2,u_2)$, we can construct a symmetrical mapping $M'$ with corresponding hypothetical search branch $b'$. We call $b'$ hypothe\-ti\-cal since it may have been pruned due to a different $H$-conflict. However, since search trees branch on all possible mappings of $u$ at the same node, $b'$ must be smaller than $b$ w.r.t. value-lexicographical order. Thus, repeated swapping of $H$-conflicting edges leads to progressively earlier hypothetical branches, eventually reaching an actual unpruned branch. This establishes completeness for the basic rule.
For the enhanced version we permit $v_2=\bot$, consider all terminal branches within the pruned subtree, and explicitly append $(\bot, u_1)$ to the branch if $u_1$ is left unmapped. This ensures that all such branches exhibit $H$-conflicts, and could thus be pruned by the basic rule.
\end{proof}

\begin{example}[Value Symmetry Breaking]
Consider graph \( H \) in \cref{fig:graphs}, vertices \( d \) and \( e \) are symmetric under positive modular symmetry. \Cref{fig:search-tree} illustrates how value symmetry pruning operates on the left side of the search tree. It is important to note that vertices \( 3 \) and \( 6 \) in \( G \) are not symmetric, yet the corresponding branch is still pruned through value symmetry. This shows that the pruning does not require symmetry among variable vertices; it applies whenever symmetric vertices exist in the value graph, without affecting the completeness of the algorithm. In \Cref{fig:search-tree}, the branch corresponding to the mapping \( M_3 = \{(0, a), (3, d), (6, e)\} \) is fully explored. The mapping \( M_5 = \{(0, a), (3, e), (6, d)\} \) is then obtained by swapping the symmetric vertices \( d \) and \( e \) in \( H \). Since both mappings produce identical bidomain partitions, that is \( P_{M_3} = P_{M_5} \), the subtrees rooted at \( M_3 \) and \( M_5 \) are result in isomorphic common subgraphs. Consequently, the subtree under \( M_5 \) can be safely skipped according to the value symmetry pruning rule.
\end{example}

\subsection{Compatibility}

The two symmetry breaking rules can be applied simultaneously during search.
This is straight-forward from an implementation perspective, and completeness of the combined pruning system can be shown using the same value-lexicographical order employed in the individual proofs. This assumes that both rules use the same ordering $<_H$ on $V(H)\cup\{\bot\}$.

\begin{theorem}[Dual Symmetry Breaking]
    \label{thm:dual-symmetry}
    Variable and value symmetry breaking rules can be applied simultaneously without violating the completeness of the algorithm.
\end{theorem}

\begin{proof}

The swap operation $(v_1,u_2)(v_2,u_1)\leftrightarrow(v_1,u_1)(v_1,u_2)$ where either $v_1,v_2$ or $u_1,u_2$ are modular symmetric induces an equivalence relation between mappings, which extends to an equivalence relation between terminal branches.
Among the branches within such an equivalence class $E$ let $b\in E$ be the smallest branch w.r.t. value-lexicographical order.
Then neither rule can prune $b$ (else $b$ would not be minimal, as demonstrated during the individual proofs), showing completeness of the combined rule system.
\end{proof}

\subsection{Algorithm Description}
\label{sec:algo}

\Cref{Alg: McSplit+Sym} outlines our proposed approach.
The algorithm begins with the main function \texttt{SymSplit}, which takes two input graphs $G$ and $H$, initializes the global variable $incumbent$ to store the best solution found, and invokes the recursive function \texttt{Search}. The core of the algorithm lies in \texttt{Search}, which explores the search space in depth-first order. At each call, if the current mapping $M$ is larger than the $incumbent$, it is recorded as the new best solution (Lines \ref{ln:best-solution}). The current size of $incumbent$ serves as a lower bound. An upper bound is then estimated based on the sum of minimum sizes over all current bidomains in the partition $P_{M}$. If this upper bound does not exceed the lower bound, the branch is pruned (Lines \ref{ln:pruning}). Otherwise, the algorithm selects a bidomain $\langle V_l, U_l \rangle$ with the smallest maximum size (Line \ref{ln:partition-selection}), and then picks a vertex $v \in V_l$ based on a vertex selection rule (Line \ref{ln:vertex-selection}), such as degree or RL-based method.
McSplit and RRSplit iterate over $U_l$ in descending order of degree, to increase the likelihood of finding a large incumbent early. We extend this to order by symmetry class and then $<_H$ in case of ties to simplify the value-symmetry check in Line~\ref{ln:val-sym-begin}.

For each candidate $w \in U_l$, the algorithm checks for $G$- and $H$-symmetry, and skips to the next vertex upon detection (Lines \ref{ln:var-sym-begin}-\ref{ln:val-sym-end}).
Otherwise, it computes updated bidomains for the new partial match $(v, w)$ (Line \ref{ln:new-partitions-computation}) and recurses on the extended match. Once all candidates in $U_l$ have been explored, the algorithm considers leaving $v$ unmatched by removing it from $V_l$. If this empties the current bidomain, the pair $\langle V_l, U_l \rangle$ is removed entirely (Lines \ref{ln:McSplit+Sym:empty-check}–\ref{ln:McSplit+Sym:empty-remove}) if $V_l$ is empty, and the search proceeds with $v$ left unmatched.

\begin{algorithm}
\caption{Symmetry-Aware BnB Algorithm}\label{Alg: McSplit+Sym}
\DontPrintSemicolon
\SetKwFunction{Main}{SymSplit}
\SetKwFunction{MCSPlitRecurse}{Search}
\SetKwProg{fn}{function}{}{}
\fn{\MCSPlitRecurse{$M, P_{M}$}}{
    \If{$|M| > |\text{incumbent}|$} {
        $\textit{incumbent} \gets M$\; \label{ln:best-solution}
    }
    $UB \gets |M| + \sum_{\langle V_l, U_l \rangle\in P_{GH}} min(|V_l|,|U_l|)$ \;
    \If{$UB \leq |\textit{incumbent}|$} {
        \textbf{return}\; \label{ln:pruning}
    }
    $\langle V_l, U_l \rangle \gets$ Smallest $max(|V_l|, |U_l|)$ over $\langle V_l, U_l \rangle \in P_M$\;\label{ln:partition-selection}
    $v \gets$ Vertex selection using degree \;\label{ln:vertex-selection}
    \For{$u \in U_l$ in order of degree/symmetry-class/$<_H$} {
        \If{$\exists\, (v', u')\in M$ with $v'\equiv v$ and $u <_H u'$\label{ln:var-sym-begin}}{
            \textbf{continue}\label{ln:var-sym-end} \tcc{Variable symmetry}
        }
        \If{$\exists\, u'\in U_l$ with $u'\equiv u$ and $u' <_H u$\label{ln:val-sym-begin}}{
            \textbf{continue}\label{ln:val-sym-end} \tcc{Value symmetry}
        }
        $P_{M'} \gets $ compute new partition bidomains \; \label{ln:new-partitions-computation}
        \MCSPlitRecurse{$M \cup \{(v, u)\}, P_{M'}$} \;
    }
    
    $V_l \gets V_l \setminus \{v\}$ \;
    \If{$|V_l| = 0$\label{ln:McSplit+Sym:empty-check}} {
        $P_{M} \gets P_{M} \setminus \{\langle V_l, U_l \rangle\}$\;\label{ln:McSplit+Sym:empty-remove}
    }
    \MCSPlitRecurse{$M\cup\{(v,\bot)\}, P_{M}$} \;
}

\fn{\Main{$G, H$}}{
    Find modular symmetry classes for $G$ and $H$\;
    \textbf{global} $\textit{incumbent} \gets \varnothing$\;
    \MCSPlitRecurse{$\varnothing, \{\langle V(G), V(H) \rangle\}$}\;
    \textbf{return} $\textit{incumbent}$ \;
}
\end{algorithm}

\subsection{Variants}\label{sec:variants}
In practice, the MCIS problem often involves variants such as graphs with self-loops or directed edges. We discuss below how the proposed symmetry-aware exact search approach extends to these cases.

\medskip
\noindent\textbf{Graphs with Self-Loops.}  
Self-loops can complicate symmetry detection. Consider two adjacent vertices \(v\) and \(u\)
with the same neighbors where only $v$ has a self-loop.
If neighborhoods are treated as sets, \(v\) and \(u\) appear structurally identical (i.e., $N^+(v)=N^+(u)$), yet they are not symmetric due to the presence of the self-loop on \(v\).
Treating neighborhoods as mutisets instead fails for the case where $v,u$ are not adjacent, have the same neighbors and both have self-loops. Then, $N^-(v)\neq N^-(u)$ and $N^+(v)\neq N^+(u)$, so symmetry detection fails.
To resolve this issue, we model self-loops by including a universal neighbor $\mathcal{O}$ in $N^-$ and $N^+$, which acts as a binary indicator to distinguish vertices with self-loops from those without.
Accordingly, we extend Definition~\ref{def:neighborhood} as follows:

\begin{definition}[Neighborhood with Self-Loops]
Let \( G = (V, E) \) possibly contain self-loops. For any \( v \in V \) with a self-loop, we define the negative and positive neighborhoods of \( v \) as,
\begin{itemize}
    \item $N^{-}(v) = \{\, u \in V \setminus \{v\} \mid (u,v) \in E \,\} \cup \{\mathcal{O}\}$,
    \item $N^{+}(v) = N^{-}(v) \cup \{v\}$.
\end{itemize}
Neighborhoods for vertices without self-loops remain unchanged.
\end{definition}

\medskip
\noindent\textbf{Directed Graphs.}  
In directed graphs, neighborhoods become pairs of sets instead of sets, distinguishing between incoming and outgoing neighbors. Here, we record a vertex as both incoming and outgoing neighbor in its positive neighborhood. This allows symmetry detection for non-adjacent vertices using $N^-$, and for neighbors with edges between them in both directions using $N^+$. Vertices with only a single directed edge between them are not symmetric and not detected as such.

We therefore extend Definition~\ref{def:neighborhood} with:

\begin{definition}[Neighborhoods in Directed Graphs]
Let \( G = (V, E) \) be a directed graph. For any \( v \in V \), we define,
\begin{itemize}
    \item $\operatorname{In}(v) = \{\, u \in V \mid (u,v) \in E \,\}$,
    \item $\operatorname{Out}(v) = \{\, u \in V \mid (v,u) \in E \,\}$.
\end{itemize}
The negative and positive neighborhoods of \( v \) are then given by
\begin{itemize}
\item $N^-(v) = \big(\;In(v),\; Out(v)\;\big)$,
\item $N^+(v) = \big(\;In(v)\cup\{v\},\; Out(v)\cup\{v\}\;\big)$.
\end{itemize}
\end{definition}

For directed graphs that may contain self-loops, these definitions can be combined by appending the symbol \( \mathcal{O} \) to \( \operatorname{In}(v) \), \( \operatorname{Out}(v) \), or both, depending on whether a loop at \( v \) is treated as incoming, outgoing, or simultaneously both. These extensions preserve all structural properties needed in our subsequent arguments. 

In addition, mutual exclusiveness (\Cref{thm:mutual-exclusiveness}) remains valid, as does the characterization of modular symmetry via graph automorphisms.

\begin{lemma}
Two vertices \( u, v \in V \) are modular symmetric if and only if the permutation that swaps \( u \) and \( v \) is an automorphism of \( G \).
\end{lemma}

\subsection{Complexity Analysis}
In this section, we discuss the data structures and implementation details used to understand the SymSplit algorithm efficiently.

\medskip
\noindent\textbf{Symmetry Identification.} Fast symmetry identification is crucial to our algorithm's performance; otherwise, the overhead of symmetry detection negates the benefits of symmetry breaking. Our symmetry detection mechanism addresses two key requirements:
\begin{itemize}
    \item Determine whether a given vertex has symmetric counterparts
    \item If symmetries exist, identify all vertices in the same symmetry class
\end{itemize}
We introduce a novel data structure to track vertex symmetries: an integer array of size $n$ (the number of vertices). Each entry stores a symmetry class identifier, enabling $O(1)$ symmetry class lookups. Non-symmetric vertices are assigned unique identifiers, making symmetry detection also an $O(1)$ operation. The space complexity of this data structure is $O(n)$.

\medskip
\noindent\textbf{Symmetry Computation.} Efficient symmetry computation is essential, as any computational overhead directly impacts the overall performance gains. We emphasize that in our experiments, symmetry computation time is included in the total MCIS solving time, making efficiency critical.
\begin{lemma}
\label{lemma:time-complexity}
The time complexity of symmetry detection is $O(n^2)$, determined by adjacency matrix traversal and hash operations.
\end{lemma}
A na\"ive approach would compare neighborhoods for all vertex pairs, requiring $O(n^3)$ time: $O(n^2)$ to compute each vertex's neighborhood, followed by pairwise comparisons. Instead, we employ a hash-based approach. For each vertex, we compute a hash value of its neighborhood using an injective hash function designed to avoid collisions on our datasets. Computing neighborhoods for all vertices requires $O(n^2)$ time. We then store these hash values in a hash map with the hash as the key and the vertex as the value. Grouping vertices by their hash values to construct the symmetry class data structure requires only $O(n)$ additional time. Since the neighborhood computation dominates the overall complexity, the total time complexity for symmetry detection is $O(n^2)$. Injective hash functions are theoretical constructs and difficult to achieve in practice when small hash values are required. However, we observe that the hash function provided by the Boost library \cite{boostlib2011} preserves injectivity across all datasets, due to hash values being sufficiently large.

\section{Experiments}\label{sec:experiments}

For our experiments, we use a Linux 13th Gen Intel(R) Core(TM) i9-13900K server with 32 CPUs and 128GB main memory. We implement the proposed algorithm in C++ similar to rest of the baselines. We compile all the algorithms using g++ version 11.4.0 with C++17.

In the experiments, we follow the same setup that has been used in all previous works of McSplit and its extensions. We give each graph pair in the dataset 1800 seconds to solve the MCIS problem. Then, we consider all the cases that cannot be solved withing 1800 seconds as \emph{hard} cases. The cases that can be solved by all the solvers within 10 seconds are considered as \emph{easy} cases and the all the cases that can be solved within 10-1800 seconds are considered \emph{moderate}. Also, note that the time taken to compute the symmetries is also included for the analysis. The implementation of our proposed algorithm can be found at \href{https://github.com/mcsolver/symsplit}{https://github.com/mcsolver/symsplit}.

\medskip
\noindent\textbf{Datasets.~}We used the following datasets in our experiments, all of which are publicly available\footnote{\url{https://perso.liris.cnrs.fr/christine.solnon/SIP.html}}. 

\begin{itemize}
    \item \textbf{BI} - Biochemical dataset \citep{2014BiochemicalReactionsDataset} contains 136 bipartite graphs with 9 to 386 vertices, representing biochemical reaction networks from \url{biomodels.net}.
    \item \textbf{CV} - The \textit{Images-CVIU11} dataset \cite{damiand2011CVIU11Dataset} includes 43 pattern graphs (22 to 151 vertices) and 146 target graphs (1,072 to 5,972 vertices), yielding 6,278 graph pairs.
    \item \textbf{LV} - The \textit{LV} dataset \cite{mccreesh2017partitioning} provides 2,352 graph pairs based on 49 pattern graphs and 48 target graphs, with vertex counts ranging from 10 to 6,671. The \textit{LargerLV} dataset contains 3,430 pairs using the same 49 pattern graphs (10 to 128 vertices) combined with 70 target graphs (138 to 6,671 vertices).
    \item \textbf{SI} - The \textit{SI} dataset \cite{zampelli2010ScaleAndSIDatasets1,zampelli2010ScaleAndSIDatasets2} includes 1,170 graph pairs, with target graphs ranging from 200 to 1,296 vertices and pattern graphs comprising 20\% to 60\% of the corresponding target graph.
\end{itemize}

\begin{table}[ht]
\centering
\caption{Percentage of vertices in modular symmetry classes.}
\vspace{-0.2cm}
\label{tab:symmetry-percentages}
\begin{tabular}{c c c}
\hline
\textbf{Dataset} & \textbf{Symmetry Percentage} & \textbf{Standard Deviation} \\
\hline
BI & 44.40\% & 21.12\% \\
LV & 21.13\% & 30.95\% \\
CV & 6.64\% & 11.36\% \\
SI & 2.60\% & 6.32\% \\
\hline
\end{tabular}
\end{table}

\medskip
\noindent\textbf{Baseline algorithms.~}In~\cite{yu2025rrsplit}, the authors report that the original McSplit algorithm~\cite{mccreesh2017partitioning} is outperformed by its reinforcement learning based extensions McSplit+RL~\cite{liu2020mcsplit+rl}, McSplit+LL~\cite{zhou2022mcsplit+ll}, and McSplit+DAL~\cite{liu2023mcsplit+dal}. Currently, the RRSplit algorithm outperforms both McSplit+DAL and the degree sequence bound~\cite{dsb2025} approaches (further details are provided in the ablation study; see Section \ref{sec:ablation}).
Therefore we compare our proposed SymSplit algorithm against the state-of-the-art RRSplit~\cite{yu2025rrsplit} algorithm.

\subsection{Performance Comparison}
In this section, we evaluate our proposed method using two key metrics: (1) \textbf{completion time}, the total time in seconds to complete all enumerations; and (2) \textbf{branch count}, the number of recursive calls made during the search.

\medskip
\noindent\textbf{Overall Performance Across Datasets.} Here we analyse the overall results of RRSplit and SymSplit over all the datasets. After removing all the easy instances from the evaluation, we see that the proposed SymSplit algorithm solves 43.98\% more instances than the RRSplit algorithm, which demonstrates the superiority of the proposed symmetry breaking technique over existing methods. We must point out that for all the datasets, SymSplit solves more cases—there is no dataset where RRSplit solves more cases. In terms of completion time, SymSplit achieves an average speedup of $75.7\times$ over RRSplit, with a maximum speedup of $1.12\times10^4$, which showcases the speedup improvement that we can gain from the proposed algorithm.

\begin{figure}[ht]
\centering
\captionsetup[subfigure]{format=hang, margin=0pt, justification=centerlast, singlelinecheck=false, skip=0pt, belowskip=10pt}
\begin{subfigure}[t]{0.48\columnwidth}
    \centering
    \includegraphics{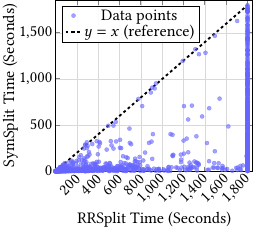}
    \caption{BI}
\end{subfigure}
\hfill
\begin{subfigure}[t]{0.48\columnwidth}
    \centering
    \includegraphics{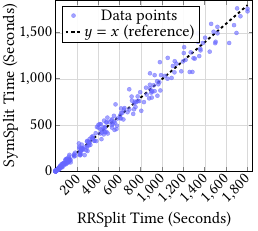}
    \caption{CV}
\end{subfigure}
\medskip
\begin{subfigure}[t]{0.48\columnwidth}
    \centering
    \includegraphics{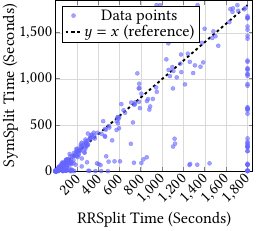}
    \caption{LV}
\end{subfigure}
\hfill
\begin{subfigure}[t]{0.48\columnwidth}
    \centering
    \includegraphics{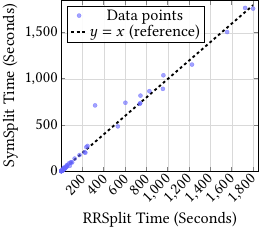}
    \caption{SI}
\end{subfigure}
\vspace{-0.4cm}
\caption{Comparing completion times across all datasets for instances solved by at least one method.}
\label{fig:completion-time-scatter}
\Description{Comparing completion times across all datasets for instances solved by at least one method.}
\end{figure}

\begin{figure}[ht]
\centering
\captionsetup[subfigure]{format=hang, margin=0pt, justification=centerlast, singlelinecheck=false, skip=0pt, belowskip=5pt}
\begin{subfigure}[t]{0.48\columnwidth}
    \centering
    \includegraphics{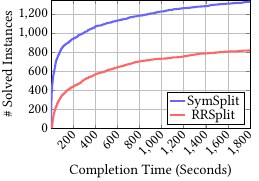}
    \caption{BI}
\end{subfigure}
\hfill
\begin{subfigure}[t]{0.48\columnwidth}
    \centering
    \includegraphics{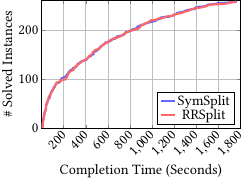}
    \caption{CV}
\end{subfigure}
\medskip
\begin{subfigure}[t]{0.48\columnwidth}
    \centering
    \includegraphics{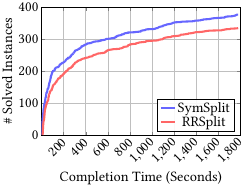}
    \caption{LV}
\end{subfigure}
\hfill
\begin{subfigure}[t]{0.48\columnwidth}
    \centering
    \includegraphics{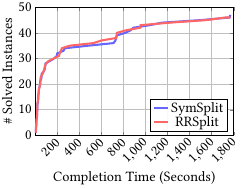}
    \caption{SI}
\end{subfigure}
\vspace{-0.4cm}
\caption{Cumulative distribution of solved instances across all datasets with respect to completion time.}
\label{fig:completion-time-cumulative}
\Description{Analysis time of Sym}
\end{figure}


\medskip
\noindent\textbf{Completion Time.} \Cref{fig:completion-time-cumulative} compares the number of solved cases by RRSplit and SymSplit with respect to completion time bounds.
The plots clearly show that for the BI dataset and LV datasets, SymSplit has clear improvements, solving more instances than RRSplit for (almost) any fixed time bound. For SI and CV datasets the improvement is marginal. This difference can be clearly explained when we analyze \cref{tab:symmetry-percentages}, where BI and LV datasets have higher percentages of modular symmetric vertices than the SI and CV datasets. \Cref{fig:completion-time-scatter} depicts the scatter plots of completion time of RRSplit versus SymSplit. All the points below the black dotted line of $y=x$ indicate cases where SymSplit takes less time, while points above the line are the ones where RRSplit performs better. We see that in BI and LV datasets the majority of points are below the line, and in fact nearly all the points in BI are below the line, which demonstrates the superior performance of SymSplit. As expected due to less symmetry for the CV and SI datasets, points are scattered along the $y=x$ lines.

\begin{figure}[ht]
\centering
\captionsetup[subfigure]{format=hang, margin=0pt, justification=centerlast, singlelinecheck=false, skip=0pt, belowskip=7pt}
\begin{subfigure}[t]{0.48\columnwidth}
    \centering
    \includegraphics{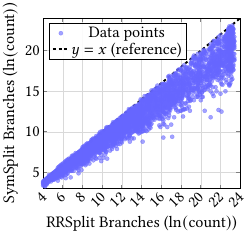}
    \caption{BI}
\end{subfigure}
\hfill
\begin{subfigure}[t]{0.48\columnwidth}
    \centering
    \includegraphics{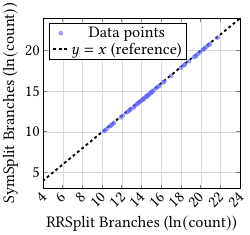}
    \caption{CV}
\end{subfigure}
\medskip
\begin{subfigure}[t]{0.48\columnwidth}
    \centering
    \includegraphics{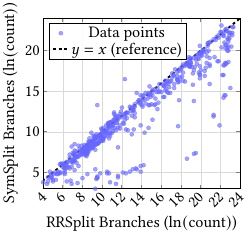}
    \caption{LV}
\end{subfigure}
\hfill
\begin{subfigure}[t]{0.48\columnwidth}
    \centering
    \includegraphics{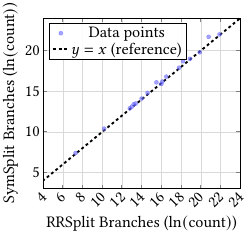}
    \caption{SI}
\end{subfigure}
\captionsetup{skip=-7pt,belowskip=7pt}
\caption{Comparing branches across all datasets for instances solved by both methods.}
\label{fig:completion-branches-scatter}
\Description{Comparing branches across all datasets for instances solved by both methods.}
\end{figure}

\begin{figure}[ht]
\captionsetup[subfigure]{format=hang, margin=0pt, justification=centerlast, singlelinecheck=false, skip=0pt, belowskip=7pt}
\centering
\begin{subfigure}[t]{0.48\columnwidth}
    \centering
    \includegraphics{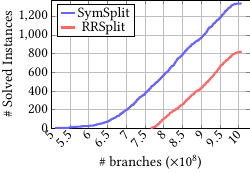}
    \caption{BI}
\end{subfigure}
\hfill
\begin{subfigure}[t]{0.48\columnwidth}
    \centering
    \includegraphics{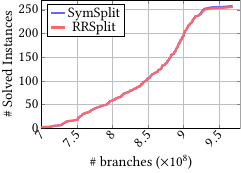}
    \caption{CV}
\end{subfigure}
\medskip
\begin{subfigure}[t]{0.48\columnwidth}
    \centering
    \includegraphics{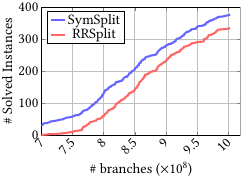}
    \caption{LV}
\end{subfigure}
\hfill
\begin{subfigure}[t]{0.48\columnwidth}
    \centering
    \includegraphics{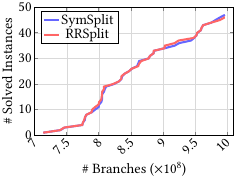}
    \caption{SI}
\end{subfigure}
\captionsetup{skip=-7pt}
\caption{Cumulative distribution of solved instances 
with respect to number of recursive calls.}
\label{fig:completion-branches-cumulative}
\Description{Cumulative distribution of solved instances across all datasets with respect to number of recursive calls.}
\end{figure}

\medskip
\noindent\textbf{Branch Count.} \Cref{fig:completion-branches-cumulative} compares the number of solved cases using bounds on the number of recursive calls, rather than completion time.
Following the completion time plots, BI and LV plots show clear improvements in solving more instances within a fixed number of branches, demonstrating the superior performance of pruning in the proposed SymSplit algorithm over the RRSplit algorithm. Following up with the previous time analysis, the SI and CV datasets show marginal improvements due to their lack of symmetries. A similar observation can be made when analysing \cref{fig:completion-branches-scatter}. The y-axis represents the number of branches taken by SymSplit while the x-axis represents the branches taken by RRSplit, both in logarithmic scale to improve visibility. In all the plots we see that nearly all the points are below the $y=x$ line except a few handful of points, which demonstrates the improvement of SymSplit over RRSplit.
We note that all symmetries identified by RRSplit are also identified by SymSplit. However, RRSplit also employs a \emph{vertex-equivalence-based upper bound}, which can result in additional bound-based pruning.

\medskip
\noindent\textbf{Locating MCIS.} Similar to the completion time, which indicates the time to completely solve the problem, and the branch count, which is the number of recursive calls to complete the search, we have investigated similar yet important measures:
\begin{itemize}
    \item Solution time: the time taken to locate the maximum common subgraph
    \item Solution branch count: the number of recursive calls to locate the maximum common subgraph
\end{itemize}
Note that even after locating this common subgraph, we have to complete all enumerations to make sure there are no larger common subgraphs. 
Locating this maximum common subgraph early is crucial because it establishes the lower bound to prune the unexplored branches if they cannot yield a better solution. Having a better lower bound ensures more pruning and leads to fast completion of the search. Upon further investigation, we note that the solution time and solution branch count both follow a similar performance improvement over all the datasets, demonstrating the superior performance of SymSplit over RRSplit.

\medskip
\noindent\textbf{Discussion.} RRSplit has some performance improvement over a few instances compared to our SymSplit. There are several factors contributing to that deviation. First, the proposed SymSplit has to compute symmetries for both graphs $G$ and $H$, which takes more time as we consider the time taken to compute the symmetries in the above evaluation. Second, during the search, we check for symmetry breaking opportunities in each iteration; however, if there are fewer opportunities, the time consumed for that check diminishes our performance. Third, their proposed method consists of a novel upper bound that contributes to pruning more branches, which provides them a slight advantage over us. We further analyse these impacts in our ablation study by applying the proposed symmetry breaking on top of RRSplit to showcase the improvement gains of each component.

\subsection{Further Analysis}
Beyond the traditional evaluation metrics such as completion time and branch count used in previous works, in this section we further analyze the performance of the SymSplit algorithm across several additional factors: (1) SymSplit performance on hard instances, (2) computational speedup gains on solved instances, and (3) effectiveness of symmetry pruning compared to upper bound-based pruning.

\begin{table}[ht]
\centering
\caption{Larger common subgraph rates for RRSplit and SymSplit with hard case statistics.}
\label{tab:hard-cases}
\begin{tabular}{l|r|r|rr}
\toprule
\multirow{2}{*}{Dataset} & \multirow{2}{*}{\shortstack{No. Hard\\Cases}} & \multirow{2}{*}{\shortstack{Non-Equal\\Rate (\%)}} & \multicolumn{2}{c}{\shortstack{Larger Common \\Subgraph Rate} (\%)} \\
\cmidrule(lr){4-5}
& & & SymSplit & RRSplit \\
\midrule
BI & 950 & 58.53 & 100.0 & 0.0 \\
CV & 4949 & 1.39 & 84.06 & 15.94 \\
LV & 4193 & 12.85 & 84.79 & 15.21 \\
SI & 731 & 3.56 & 92.31 & 7.69 \\
\midrule
All & 10823 & 11.0 & 92.02 & 7.98\\
\bottomrule
\end{tabular}
\end{table}

\begin{figure}[ht]
\centering
\resizebox{0.8\columnwidth}{!}{\includegraphics{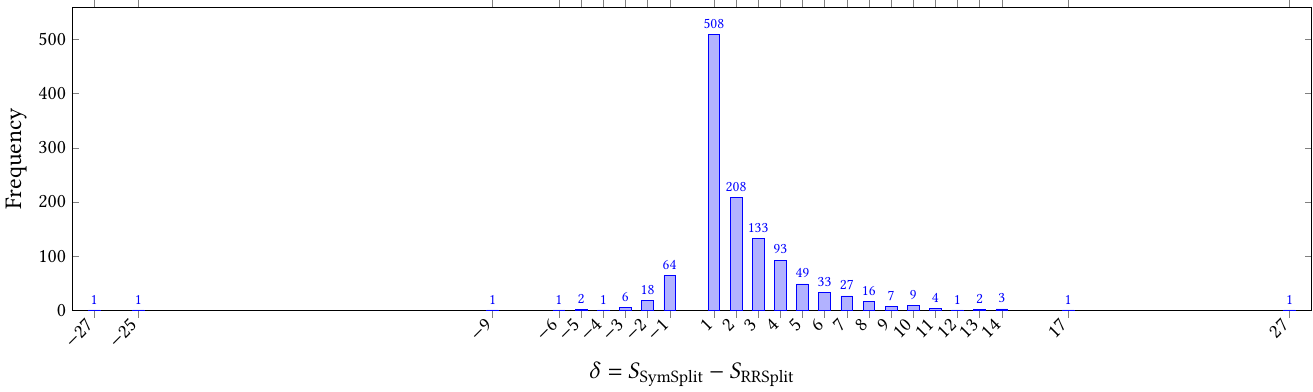}}
\caption{Frequency distribution of $\delta$ values.}
\Description{Frequency distribution of $\delta$ values.}
\label{fig:delta-frequencies}
\end{figure}

\medskip
\noindent\textbf{Hard Instances.} 
Instances that are not solved by RRSplit or SymSplit within the time limit are considered hard cases. Even when exact solutions are not found, identifying a larger common subgraph within the available time can still provide useful approximations and insights for real-world applications \citep{marini2005exact}. In these hard instances, there are cases where both RRSplit and SymSplit find equal-sized common subgraphs, and cases where one algorithm finds a larger common subgraph than the other.

\Cref{tab:hard-cases} first shows the number of hard cases in each dataset: BI has 950 graph pairs as hard instances, CV has 4,949, LV has 4,193, and SI has 731. Among these instances, many return equal-sized common subgraphs for both RRSplit and SymSplit. The non-equal percentages are shown in the next column: 58.53\% of hard cases in BI return different-sized larger common subgraphs between RRSplit and SymSplit, as do 1.39\% of hard instances in CV, 12.85\% in LV, and 3.56\% in SI.

\Cref{tab:hard-cases} then compares these non-equal hard instances using the measure called \emph{Larger Common Subgraph Rate}, which reflects the percentage of cases where RRSplit or SymSplit finds a larger common subgraph than the other, excluding equal cases. We see that in more than 84\% of these cases, the SymSplit algorithm finds a better common subgraph than RRSplit across all datasets. In particular, in the BI dataset, SymSplit always finds a better common subgraph while RRSplit does not outperform in a single case. For the CV dataset, the larger common subgraph rate of SymSplit is 84.06\%; for LV it is 84.79\%; and for SI it is 92.31\%. This showcases the applicability of SymSplit over RRSplit for more difficult problem instances of MCIS and for finding reasonably good approximations.

We further analyze the gaps in these unequal hard instances by computing $\delta = S_{\text{SymSplit}} - S_{\text{RRSplit}}$, where $S_{\text{SymSplit}}$ is the common subgraph size returned by SymSplit and $S_{\text{RRSplit}}$ is the common subgraph size returned by RRSplit for a given instance. A positive $\delta$ indicates that SymSplit returns a better subgraph, while a negative $\delta$ indicates otherwise. Our initial expectation was that this delta value should be close to 1. Even finding a subgraph with just one additional vertex contributes significantly to pruning more branches, as it acts as the lower bound. Finding subgraphs with more than one additional matching is extremely difficult in NP-hard combinatorial problems. However, to our surprise, when we compute the mean of these delta values across all datasets, we obtain a mean of 2.17 with a standard deviation of 2.75. This shows that the larger subgraphs found by SymSplit do not contain just one additional vertex pair, but can more values than 2.

The \Cref{fig:delta-frequencies} shows the $\delta$ frequencies across the all the hard instances in all datasets. As expected majority of the cases have the difference of 1. There are very few cases in the negative side which shows the cases where RRSplit has a larger subgraph than SymSplit. However, the majority is in the positive side showing the effectiveness of SymSplit in finding larger common subgraphs. Unexpectedly, in the positive side also we have large distribution of values with values ranging from 1-14 and some outliers with 17 and 27 with only one instance. Overall, this showcases that SymSplit is able to find larger common subgraphs in many hard instances even with the gap of more than 1.

\begin{table}[!ht]
    \centering
    \caption{Comparing average and maximum speedup achieved by SymSplit over RRSplit algorithm.}
    \label{tab:speedup}
    \begin{tabular}{@{}l|c|c@{}}
        \toprule
        Dataset & Average Speedup & Maximum Speedup \\
        \midrule
        BI     & $\times 51.1$ & $\times 8.06e3$ \\
        LV     & $\times 253$ & $\times 1.12e4$ \\
        CV     & $\times 1.12$ & $\times 1.72$ \\
        SI     & $\times 1.09$ & $\times 1.37$ \\
        \midrule
        All    & $\times 75.7$ & $\times1.12e4$ \\
        \bottomrule
    \end{tabular}
\end{table}

\medskip
\noindent\textbf{Computational Speedup.~}\Cref{tab:speedup} presents the average and maximum speedups achieved by SymSplit over RRSplit across all datasets. As expected, the CV and SI datasets yielded the lowest speedups due to limited symmetry availability. Nevertheless, the proposed method achieved modest improvements, with average speedups of $1.09\times$ (SI) and $1.12\times$ (CV), and maximum speedups reaching $1.37\times$ and $1.72\times$, respectively. While SymSplit demonstrated the largest and most consistent performance improvements on the BI dataset in terms of completion time, branching, and solved instances (as shown in figures~\ref{fig:completion-time-scatter}-\ref{fig:completion-branches-cumulative}), the LV dataset achieved the highest speedup gains: an average speedup of $253\times$ and a maximum speedup of $11,200\times$. In comparison, the BI dataset yielded an average speedup of $51.5$ and a maximum of $8,060$.

This counterintuitive result can be explained by the symmetry distribution characteristics. Although the LV dataset has lower average symmetry than BI, it exhibits substantially higher standard deviation in symmetry. Consequently, LV contains a wider range of graph structures, including instances with exceptionally high symmetries alongside others with minimal or no symmetry. The presence of these highly symmetric instances enables the dramatic speedup gains observed with SymSplit on the LV dataset.

\begin{table}[h]
\centering
\begin{tabular}{l |c | c | c}
\hline
Dataset & Mean & Stdev & High sym pruning cases (\%) \\
\hline
BI & 26.0 & 19.31 & 0.29 \\
LV & 149.57 & 1068.74 & 8.99 \\
SI & 0.04 & 0.33 & 0 \\
CV & 0.17 & 0.09 & 0 \\
\midrule
All & 49.15 & 544.65 & 2.43
\end{tabular}
\caption{Effectiveness in pruning dataset wise}
\Description{Effectiveness in pruning dataset wise}
\label{tab:pruning-effectiveness}
\end{table}

\medskip
\noindent\textbf{Effectiveness of Pruning Strategies.}
Our proposed method integrates two pruning strategies: a bounding rule inherited from McSplit and a novel symmetry-breaking framework. To assess their respective contributions, we analyzed the number of branches pruned by each approach across all instances. As shown in \Cref{tab:pruning-effectiveness}, on average, for every $100$ branches pruned by the bounding rule, the symmetry-breaking rule prunes an additional $49.15$ branches, with a standard deviation of $\sigma = 544.65$. This suggests that symmetry breaking eliminates approximately one branch for every two pruned by bounding, highlighting its significant impact. Notably, in the LV dataset, symmetry breaking achieves even greater effectiveness, removing $149.67$ branches on average for every $100$ pruned by bound-based pruning.

The high standard deviation ($\sigma = 544.65$) indicates substantial variability in the relative effectiveness of the two strategies. In certain cases, symmetry pruning eliminates far more branches than bounding-rule pruning. To quantify this, we compare the percentage of instances where symmetry pruning surpasses bound-based pruning, referred to as \emph{high symmetry pruning cases}. We observe that in 8.99\% of cases in the LV dataset, symmetry pruning is more effective than bound-based pruning. In the BI dataset, this figure is only 0.29\%. Overall, symmetry breaking outperforms bounding in $2.43\%$ of all instances. Notably, the LV dataset sees more branches pruned via symmetry-breaking than via bound-based pruning. The previously discussed high standard deviation in the LV dataset’s symmetry distribution further explains this observation.

\begin{figure*}[ht]
\begin{minipage}[t]{0.47\textwidth}
\includegraphics{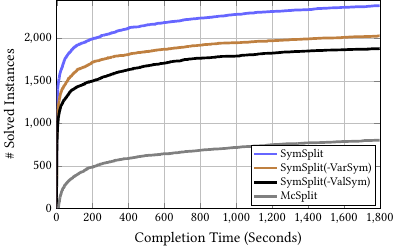}
\caption{Comparison of SymSplit against its individual components, using only variable symmetry breaking or only value symmetry breaking.}
\label{fig:ablation-symSplit}
\Description{Analysis time of Sym}
\end{minipage}
\hfill
\begin{minipage}[t]{0.47\textwidth}
\includegraphics{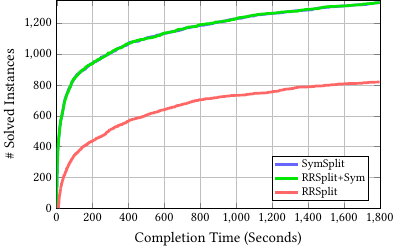}
\caption{Performance comparison of RRSplit, RRSplit+Sym, and SymSplit by number of solved instances over completion time. The SymSplit and RRSplit+Sym curves overlap due to nearly identical performance.}
\label{fig:ablation-rrsplit}
\Description{Analysis time of Sym}
\end{minipage}
\end{figure*}

\subsection{Ablation Study}\label{sec:ablation}
In this section, we analyze the effectiveness of each component in detail to demonstrate their individual contributions to the overall algorithm performance.

\medskip
\noindent\textbf{Variable and Value Symmetry Breaking.~}\Cref{fig:ablation-symSplit} shows the effectiveness of the final algorithm with each symmetry breaking component that we proposed. The y-axis of the plot shows the number of cases solved and the x-axis shows the completion time in seconds. Here we analyze three algorithms:
\begin{itemize}
    \item SymSplit: the proposed algorithm with both variable and value symmetry breaking rules
    \item SymSplit(-VarSym): the proposed algorithm SymSplit without the variable symmetry breaking rule
    \item SymSplit(-ValSym): the proposed SymSplit algorithm without the value symmetry breaking rule
    \item McSplit: the McSplit algorithm which does not include variable or value symmetries.
\end{itemize}

As expected, the SymSplit algorithm performs better than the other variants. By analyzing \cref{fig:ablation-symSplit}, we can see that the drop in the number of cases solved by both SymSplit(-VarSym) and SymSplit(-ValSym) demonstrates the effectiveness of variable and value symmetry breaking individually.

We also observe that SymSplit(-ValSym) solves fewer cases than SymSplit(-VarSym), which indicates that value symmetry breaking has more impact than variable symmetry breaking. However, this difference does not diminish the importance of variable symmetry breaking, as removing it also significantly reduces the number of cases solved. These results demonstrate the effectiveness of incorporating both variable and value symmetry breaking rules in the algorithm. Even when comparing with the McSplit algorithm without both variable and value symmetry breaking we see how much impact that we can gain from both value and variable symmetry breaking together.

\medskip
\noindent\textbf{Comparison with RRSplit.~}The proposed symmetry breaking can also be used as a plug-in for existing algorithms without violating completeness or correctness. \Cref{fig:ablation-rrsplit} compares three algorithms:
\begin{itemize}
    \item RRSplit: the algorithm proposed by \cite{yu2025rrsplit} 
    \item RRSplit+Sym: the proposed symmetry breaking as a plug-in for the RRSplit algorithm
    \item SymSplit: the proposed symmetry breaking algorithm
\end{itemize}
We compare the number of solved cases on the y-axis and time in seconds on the x-axis. The results clearly show the performance gain obtained by applying the proposed symmetry breaking on top of the RRSplit algorithm, demonstrating the superiority of the proposed approach. \Cref{fig:ablation-rrsplit} shows that RRSplit+Sym and SymSplit perform nearly identically. The only difference between RRSplit+Sym and SymSplit is the vertex equivalence-based bound. When both variable and value symmetry breaking are active, RRSplit's specialized upper bound provides minimal additional benefit. This suggests complete symmetry breaking may reduce the need for complex bounding heuristics. We therefore implement SymSplit without RRSplit's upper bound to reduce implementation complexity. Users requiring maximum performance on specific graph classes may benefit from combining all techniques.

\begin{figure}[ht]
\includegraphics{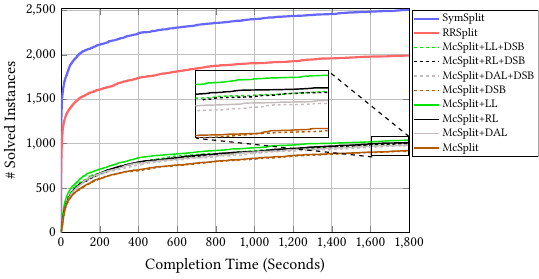}
\caption{
Performance comparison of SymSplit and RRSplit with the McSplit and its RL-based extensions (McSplit, McSplit+RL, McSplit+LL, and McSplit+DAL).}
\label{fig:ablation-rl-dsb}
\Description{Analysis time of Sym}
\end{figure}

\medskip
\noindent\textbf{Comparison with RL and DSB Extensions.~} 
In \Cref{fig:ablation-rl-dsb}, we compare the performance of SymSplit and RRSplit against McSplit, its RL extensions (McSplit+RL, McSplit+LL, and McSplit+DAL), and their respective DSB extensions. We observe that, compared to McSplit, the RL extensions offer some performance improvement. However, this improvement is extremely small relative to the performance gains achieved by the symmetry breaking techniques in RRSplit and SymSplit. This highlights the scale of the performance advantage provided by symmetry breaking techniques.

\section{Related Work}
The MCIS problem has inspired a range of different algorithms. 
Existing work can broadly be classified into algorithmic and reinforcement 
learning-based improvements to the McSplit framework. 
\Cref{Tab: Related Works} summarizes representative algorithms from the 
literature and their key capabilities: variable symmetry exploitation, value 
symmetry exploitation, custom bound computation, learning-based vertex 
selection, and exactness guarantees. To the best of our knowledge, our approach is the first MCIS-specific exact method to simultaneously exploit symmetries in both input graphs. While dual-side symmetry breaking has been studied in the constraint programming literature for related problems, our contribution lies in applying this principle within an exact maximum common induced subgraph solver.

\begin{table}[ht]
\centering
\caption{Summary of MCIS algorithms.}
\small
\begin{tabular}{@{}lccccc@{}}
\toprule
\textbf{Method} & \textbf{Var.} & \textbf{Val.} & \textbf{Custom} & \textbf{Learn.} & \textbf{Exact} \\
                & \textbf{Symmetry} & \textbf{Symmetry} & \textbf{Bound} & \textbf{Based} & \textbf{} \\
\midrule
Exhaustive \cite{akutsu2012complexity}        & \xmark & \xmark & \xmark & \xmark & \cmark \\
McSplit \cite{mccreesh2017partitioning}        & \xmark & \xmark & \xmark & \xmark & \cmark \\
RL \cite{liu2020mcsplit+rl, zhou2022mcsplit+ll, liu2023mcsplit+dal}     & \xmark & \xmark & \xmark & \cmark & \cmark \\
DSB \cite{dsb2025}     & \xmark & \xmark & \cmark & \cmark & \cmark \\
Trummer \cite{trummer2021engineering}       & \cmark & \cmark & \xmark & \xmark & \xmark \\
RRSplit \cite{yu2025rrsplit}      & \cmark & \xmark & \cmark & \xmark & \cmark \\
\hline
\textbf{SymSplit} & \cmark & \cmark & \xmark & \xmark & \cmark \\
\bottomrule
\end{tabular}
\label{Tab: Related Works}
\end{table}

\vspace{0.15cm}
\noindent\textbf{Algorithmic Improvements.~}
The foundational work by Akutsu and Tamura \cite{akutsu2012complexity} provides a theoretical analysis of the maximum common subgraph problem and several of its variants. The authors establish NP-hardness results for multiple problem formulations, including the maximum common connected edge subgraph and maximum common connected induced subgraph problems, even under certain bounded-degree conditions. While this work establishes important complexity-theoretic foundations, the proposed algorithm employs a na\"ive enumeration approach: it exhaustively enumerates all subgraphs in both $G$ and $H$, then tests each pair for isomorphism. This enumeration-based strategy is fundamentally inefficient compared to modern branch-and-bound methods that employ advanced pruning.

The McSplit algorithm~\cite{mccreesh2017partitioning} introduced an efficient branch-and-bound framework that partitions unmatched vertices of the two graphs into bidomains based on structural compatibility. Its upper bound, defined in Equation~\ref{Eqn: McSplitBound}, enables aggressive pruning while preserving completeness. Building on this foundation, Yu et~al.~\cite{yu2025rrsplit} proposed the RRSplit algorithm, which refines McSplit by introducing vertex-equivalence-based reductions and a tighter upper bound that accounts for structurally equivalent vertices. These improvements eliminate redundant exploration of isomorphic subtrees and strengthen pruning. However, RRSplit’s equivalence reasoning is confined to local vertex redundancy within the variable graph and does not account for symmetries in the value graph $H$.

Following a similar line of research, Trummer et al.~\cite{trummer2021engineering} investigated the application of symmetry breaking in the McSplit framework by exploiting graph automorphisms. While their method effectively reduced redundant exploration in practice, it introduced approximations that could compromise the exactness of the final solution. As such, designing sound and complete symmetry-breaking techniques for the MCIS problem remains an open challenge.

\vspace{0.15cm}
\noindent\textbf{Reinforcement Learning-Based Improvements.} Recent learning-based extensions of McSplit have focused on improving the branching strategy to enhance search efficiency. Liu et al.~\cite{liu2020mcsplit+rl} proposed McSplit+RL, which replaces McSplit’s hand-crafted branching rule with a reinforcement learning (RL) policy that prioritizes vertex pairs based on their expected contribution to reducing the upper bound defined in Equation~\ref{Eqn: McSplitBound}. Two value functions, $S_p(v)$ for vertices in $V(G)$ and $S_t(u)$ for vertices in $V(H)$, are learned to guide the selection of vertex pairs that are likely to yield the largest bound reduction.

Building on this idea, Zhou et al.~\cite{zhou2022mcsplit+ll} observed that the accumulated reward history in McSplit+RL could introduce bias during learning. They proposed McSplit+LL, which separates short-term rewards for selecting a vertex $v \in V(G)$ from long-term rewards for mapping $v$ to a vertex $u \in V(H)$. This distinction improves learning stability and matching quality. They further introduced a mechanism for jointly mapping multiple leaf vertices that exhibit equivalent structural patterns, thereby reducing search redundancy without compromising optimality.
Liu et al.~\cite{liu2023mcsplit+dal} extended this framework by integrating structural information from the bidomain partitioning $P_M^k$ into the reward function. Unlike earlier methods that focused exclusively on minimizing the upper bound, their Domain Action Learning (DAL) approach promotes selections that reduce the number of bidomains $|P_M^k|$, encouraging broader and more diverse exploration of the search space. To combine the advantages of both strategies, they proposed a hybrid model that balances bound-focused and structure-aware learning. 

More recently, the Degree Sequence Bound (DSB) framework~\cite{dsb2025} introduced a stronger upper bound for McSplit and its RL-based extensions by exploiting degree sequence information. Although this bound is computationally more demanding, it provides a tighter estimate of the remaining search potential. To offset the additional cost, the authors employed an RL model that selectively applies the bound only in promising regions of the search space, achieving a balance between theoretical precision and computational efficiency.

\begin{figure}[ht]
    \centering
    \includegraphics{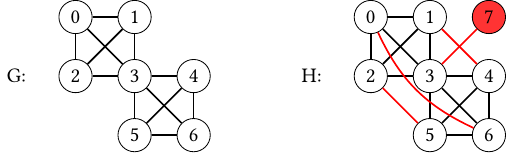}
    \caption{Example graphs $G$ and $H$ differentiating partial graph alignment (PGA) from the MCIS problem.}
    \label{fig:pga}
    \Description{PGA Example}
\end{figure}

\vspace{0.15cm}
\noindent\textbf{Other Graph-Matching Problems:~} Instead of MCIS, other similarity metrics are used to compare two graphs, though the resulting problems can be fundamentally different.
One such problem is the subgraph isomorphism (SI) problem, which attempts to find a pattern graph $G$ inside a target graph $H$. SI has two variants: one seeks an induced subgraph of $H$ isomorphic to $G$, while the other seeks a subgraph of $H$ (not necessarily induced) isomorphic to $G$ \cite{hoffmann2017simethod}. On the other hand, the graph alignment problem seeks a vertex-to-vertex correspondence between the vertices of $G$ and $H$ that maximizes edge overlap. When the sizes of $G$ and $H$ are equal, the problem is called the graph alignment problem; otherwise, it is called the partial graph alignment (PGA) problem \cite{petsinis2025alpine}.

The MCIS problem differs from these in two major ways.
First, the SI problem and the graph alignment problem both have a predefined pattern graph to be matched against the target graph, whereas the MCIS problem does not have such predetermined graphs; instead, both subgraphs in MCIS are determined during the search. Additionally, while the graph alignment problem aims to maximize alignment, the MCIS problem must satisfy isomorphism at each step of the search, with the goal of maximizing the size of the isomorphic subgraphs. The following example show the difference between the partial graph alignment and the MCIS problem.

\begin{example}
Consider graphs $G$ and $H$ in Figure~\ref{fig:pga}, where $G$ is a subgraph of $H$ with differences highlighted in red. Any optimal solution to PGA would leave vertex $7$ in $H$ unmatched. However, a MCIS between $G$ and $H'=H\setminus\{7\}$ contains only $4$ vertices, while a MCIS between $G$ and $H$ is of size $5$.
\end{example}

\section{Conclusion}\label{sec:conclusion}

This paper addresses the challenge of efficiently reducing redundant exploration in branch-and-bound algorithms by exploiting structural symmetries in the Maximum Common Induced Subgraph (MCIS) problem. While symmetry breaking has been widely used in combinatorial search, it has received limited attention for the MCIS problem. We propose a polynomial-time symmetry-breaking framework that identifies and prunes symmetric subtrees during the search process. Unlike previous methods, our approach simultaneously accounts for both variable and value symmetries, ensuring a more comprehensive elimination of redundant branches. Experimental results demonstrate that our approach consistently reduces computation time and the number of recursive calls, achieving substantial improvements on standard benchmarks.

\section{Limitation and Future Work}

The proposed symmetry-breaking rules effectively prune redundant search branches by identifying those that have lexicographically smaller symmetrical counterparts. However, they do not eliminate all such cases and may still produce multiple mappings belonging to the same symmetry class.

\begin{example}
Consider the following two search branches:
\begin{align*}
b &= (v_0,u_1)(v_1,u_0)(v_2,u_2),\\
b' &= (v_0,u_0)(v_1,u_2)(v_2,u_1),
\end{align*}
where \( v_0 \equiv v_2 \), \( u_0 \equiv u_2 \), and \( u_0 < u_1 < u_2 \).
The two branches are symmetrical yet remain conflict-free under the current pruning rules.
\end{example}

Detecting such cases can be achieved with a single pass over the sequence of \((v_i, u_i)\) pairs in branch order, and is only needed when both newly matched vertices belong to non-trivial equivalence classes. In practice, these scenarios occur infrequently, so extending the pruning rules to capture them is unlikely to yield substantial performance gains. 

For future work, we plan to more tightly integrate symmetry reasoning into the search process, potentially combining symmetry detection with reinforcement learning. Such integration could enable adaptive pruning strategies that further reduce redundant exploration while preserving exact solutions. Additionally, applying similar symmetry-breaking rules to other combinatorial search problems is a promising direction. 
\begin{acks}

The authors are grateful to Professor Qing Wang for insightful discussions and valuable feedback. This work was supported in part by the Australian Research Council (Discovery Project DP210102273). We also thank the anonymous reviewers for their constructive suggestions, which improved the quality of this paper. Generative AI tools, including Claude AI and Grammarly, were used solely for text editing and figure generation; all intellectual contributions are the original work of the authors.

\end{acks}

\bibliographystyle{ACM-Reference-Format}
\bibliography{sample-base}

@String{Computing = "Computing" }

@String{Computer = "{IEEE} Computer" }

@String{Springer = "Springer-Verlag" }

@inproceedings{mccreesh2017partitioning,
author = {McCreesh, Ciaran and Prosser, Patrick and Trimble, James},
title = {A partitioning algorithm for maximum common subgraph problems},
year = {2017},
isbn = {9780999241103},
publisher = {AAAI Press},
booktitle = {Proceedings of the 26th International Joint Conference on Artificial Intelligence},
pages = {712–719},
numpages = {8},
location = {Melbourne, Australia},
address = {Melbourne, Australia},
series = {IJCAI'17}
}

@article{liu2020mcsplit+rl,
author = {Liu, Yanli and Li, Chu-Min and Jiang, Hua and He, Kun},
year = {2020},
month = {04},
pages = {2392-2399},
title = {A Learning Based Branch and Bound for Maximum Common Subgraph Related Problems},
volume = {34},
journal = {Proceedings of the AAAI Conference on Artificial Intelligence},
doi = {10.1609/aaai.v34i03.5619}
}

@inproceedings{zhou2022mcsplit+ll,
  title     = {A Strengthened Branch and Bound Algorithm for the Maximum Common (Connected) Subgraph Problem},
  author    = {Zhou, Jianrong and He, Kun and Zheng, Jiongzhi and Li, Chu-Min and Liu, Yanli},
  booktitle = {Proceedings of the Thirty-First International Joint Conference on
               Artificial Intelligence, {IJCAI-22}},
  publisher = {International Joint Conferences on Artificial Intelligence Organization},
  editor    = {Lud De Raedt},
  pages     = {1908--1914},
  year      = {2022},
  month     = {7},
  note      = {Main Track},
  doi       = {10.24963/ijcai.2022/265},
  url       = {https://doi.org/10.24963/ijcai.2022/265},
  address   = {Messe Wien, Vienna, Austria},
}

@article{2014BiochemicalReactionsDataset,
title = {On the subgraph epimorphism problem},
journal = {Discrete Applied Mathematics},
volume = {162},
pages = {214-228},
year = {2014},
issn = {0166-218X},
doi = {https://doi.org/10.1016/j.dam.2013.08.008},
url = {https://www.sciencedirect.com/science/article/pii/S0166218X13003491},
author = {Steven Gay and François Fages and Thierry Martinez and Sylvain Soliman and Christine Solnon}
}

@article{damiand2011CVIU11Dataset,
title = {Polynomial algorithms for subisomorphism of nD open combinatorial maps},
journal = {Computer Vision and Image Understanding},
volume = {115},
number = {7},
pages = {996-1010},
year = {2011},
note = {Special issue on Graph-Based Representations in Computer Vision},
issn = {1077-3142},
doi = {https://doi.org/10.1016/j.cviu.2010.12.013},
url = {https://www.sciencedirect.com/science/article/pii/S1077314211000816},
author = {Guillaume Damiand and Christine Solnon and Colin {de la Higuera} and Jean-Christophe Janodet and Émilie Samuel}
}

@article{zampelli2010ScaleAndSIDatasets1,
author = {Zampelli, Stéphane and Deville, Yves and Solnon, Christine},
year = {2010},
month = {07},
pages = {327-353},
title = {Solving subgraph isomorphism problems with constraint programming},
volume = {15},
journal = {Constraints},
doi = {10.1007/s10601-009-9074-3}
}

@article{zampelli2010ScaleAndSIDatasets2,
title = {AllDifferent-based filtering for subgraph isomorphism},
journal = {Artificial Intelligence},
volume = {174},
number = {12},
pages = {850-864},
year = {2010},
issn = {0004-3702},
doi = {https://doi.org/10.1016/j.artint.2010.05.002},
url = {https://www.sciencedirect.com/science/article/pii/S0004370210000718},
author = {Christine Solnon}
}

@article{liu2023mcsplit+dal, 
title={Hybrid Learning with New Value Function for the Maximum Common Induced Subgraph Problem}, 
volume={37}, url={https://ojs.aaai.org/index.php/AAAI/article/view/25519}, 
DOI={10.1609/aaai.v37i4.25519}, 
abstractNote={Maximum Common Induced Subgraph (MCIS) is an important NP-hard problem with wide real-world applications. An efficient class of MCIS algorithms uses Branch-and-Bound (BnB), consisting in successively selecting vertices to match and pruning when it is discovered that a solution better than the best solution found so far does not exist. The method of selecting the vertices to match is essential for the performance of BnB. In this paper, we propose a new value function and a hybrid selection strategy used in reinforcement learning to define a new vertex selection method, and propose a new BnB algorithm, called McSplitDAL, for MCIS. Extensive experiments show that McSplitDAL significantly improves the current best BnB algorithms, McSplit+LL and McSplit+RL. An empirical analysis is also performed to illustrate why the new value function and the hybrid selection strategy are effective.}, 
number={4}, 
journal={Proceedings of the AAAI Conference on Artificial Intelligence}, 
author={Liu, Yanli and Zhao, Jiming and Li, Chu-Min and Jiang, Hua and He, Kun}, 
year={2023}, month={Jun.}, pages={4044-4051} }

@article{kawabata2014chemapplication,
author = {Kawabata, Takeshi and Nakamura, Haruki},
title = {3D Flexible Alignment Using 2D Maximum Common Substructure: Dependence of Prediction Accuracy on Target-Reference Chemical Similarity},
journal = {Journal of Chemical Information and Modeling},
volume = {54},
number = {7},
pages = {1850-1863},
year = {2014},
doi = {10.1021/ci500006d},
    note ={PMID: 24895842},
URL = { 
        https://doi.org/10.1021/ci500006d  
},
eprint = { https://doi.org/10.1021/ci500006d}
}

@InProceedings{ndiaye2011cpmethod,
author="Ndiaye, Samba Ndojh
and Solnon, Christine",
editor="Lee, Jimmy",
title="CP Models for Maximum Common Subgraph Problems",
booktitle="Principles and Practice of Constraint Programming -- CP 2011",
year="2011",
publisher="Springer Berlin Heidelberg",
address="Berlin, Heidelberg",
pages="637--644",
isbn="978-3-642-23786-7"
}

@article{hoffmann2017simethod, 
title={Between Subgraph Isomorphism and Maximum Common Subgraph}, 
volume={31}, url={https://ojs.aaai.org/index.php/AAAI/article/view/11137}, DOI={10.1609/aaai.v31i1.11137}, 
abstractNote={ &lt;p&gt; When a small pattern graph does not occur inside a larger target graph, we can ask how to find &quot;as much of the pattern as possible&quot; inside the target graph. In general, this is known as the maximum common subgraph problem, which is much more computationally challenging in practice than subgraph isomorphism. We introduce a restricted alternative, where we ask if all but k vertices from the pattern can be found in the target graph. This allows for the development of slightly weakened forms of certain invariants from subgraph isomorphism which are based upon degree and number of paths.  We show that when k is small, weakening the invariants still retains much of their effectiveness. We are then able to solve this problem on the standard problem instances used to benchmark subgraph isomorphism algorithms, despite these instances being too large for current maximum common subgraph algorithms to handle. Finally, by iteratively increasing k, we obtain an algorithm which is also competitive for the maximum common subgraph&lt;br /&gt;&lt;br /&gt; &lt;/p&gt; }, 
number={1}, journal={Proceedings of the AAAI Conference on Artificial Intelligence}, 
author={Hoffmann, Ruth and McCreesh, Ciaran and Reilly, Craig}, 
year={2017}, 
month={02}
}

@article{mcgregor1982backtracksa,
  title={Backtrack search algorithms and the maximal common subgraph problem},
  author={McGregor, James J},
  journal={Software: Practice and Experience},
  volume={12},
  number={1},
  pages={23--34},
  year={1982},
  publisher={Wiley Online Library}
}

@InProceedings{vismara2008cpmethod,
author="Vismara, Philippe
and Valery, Beno{\^i}t",
editor="Le Thi, Hoai An
and Bouvry, Pascal
and Pham Dinh, Tao",
title="Finding Maximum Common Connected Subgraphs Using Clique Detection or Constraint Satisfaction Algorithms",
booktitle="Modelling, Computation and Optimization in Information Systems and Management Sciences",
year="2008",
publisher="Springer Berlin Heidelberg",
address="Berlin, Heidelberg",
pages="358--368",
isbn="978-3-540-87477-5"
}

@article{levi1973note,
  title={A note on the derivation of maximal common subgraphs of two directed or undirected graphs},
  author={Levi, Giorgio},
  journal={Calcolo},
  volume={9},
  number={4},
  pages={341--352},
  year={1973}
}

@article{balas1986finding,
  title={Finding a maximum clique in an arbitrary graph},
  author={Balas, Egon and Yu, Chang Sung},
  journal={SIAM Journal on Computing},
  volume={15},
  number={4},
  pages={1054--1068},
  year={1986}
}

@article{koch2001enumerating,
  title={Enumerating all connected maximal common subgraphs in two graphs},
  author={Koch, Ina},
  journal={Theoretical Computer Science},
  volume={250},
  number={1-2},
  pages={1--30},
  year={2001}
}

@article{ehrlich2011chemanalisis,
author = {Ehrlich, Hans-Christian and Rarey, Matthias},
title = {Maximum common subgraph isomorphism algorithms and their applications in molecular science: a review},
journal = {WIREs Computational Molecular Science},
volume = {1},
number = {1},
pages = {68-79},
year = {2011}
}

@inproceedings{fahle2001symmetry,
author = {Fahle, Torsten and Schamberger, Stefan and Sellmann, Meinolf},
title = {Symmetry Breaking},
year = {2001},
isbn = {3540428631},
publisher = {Springer-Verlag},
address = {Berlin, Heidelberg},
booktitle = {Proceedings of the 7th International Conference on Principles and Practice of Constraint Programming},
pages = {93–107},
numpages = {15},
series = {CP '01}
}

@phdthesis{trummer2021engineering,
  title={Engineering Maximum Common Subgraph Algorithms for Large Graphs},
  author={Trummer, Jonathan},
  year={2021},
  school={University of Vienna}
}

@InProceedings{akutsu2012complexity,
author="Akutsu, Tatsuya
and Tamura, Takeyuki",
editor="Chao, Kun-Mao
and Hsu, Tsan-sheng
and Lee, Der-Tsai",
title="On the Complexity of the Maximum Common Subgraph Problem for Partial k-Trees of Bounded Degree",
booktitle="Algorithms and Computation",
year="2012",
publisher="Springer Berlin Heidelberg",
address="Berlin, Heidelberg",
pages="146--155",
isbn="978-3-642-35261-4"
}

@InProceedings{marini2005exact,
author="Marini, Simone
and Spagnuolo, Michela
and Falcidieno, Bianca",
editor="Brun, Luc
and Vento, Mario",
title="From Exact to Approximate Maximum Common Subgraph",
booktitle="Graph-Based Representations in Pattern Recognition",
year="2005",
publisher="Springer Berlin Heidelberg",
address="Berlin, Heidelberg",
pages="263--272",
isbn="978-3-540-31988-7"
}

@inproceedings{lahiri2008mining-network,
author = {Lahiri, Mayank and Berger-Wolf, Tanya Y.},
title = {Mining Periodic Behavior in Dynamic Social Networks},
year = {2008},
isbn = {9780769535029},
publisher = {IEEE Computer Society},
address = {USA},
url = {https://doi.org/10.1109/ICDM.2008.104},
doi = {10.1109/ICDM.2008.104},
booktitle = {Proceedings of the 2008 Eighth IEEE International Conference on Data Mining},
pages = {373–382},
numpages = {10},
series = {ICDM '08}
}

@article{vijayalakshmi2011performance-network,
  title={Performance monitoring of large communication networks using maximum common subgraphs},
  author={Vijayalakshmi, R and Nadarajan, R and Nirmala, P and Thilaga, M},
  journal={International Journal of Artificial Intelligence},
  volume={6},
  number={S11},
  pages={72--86},
  year={2011}
}

@article{yu2025rrsplit,
  title={Fast Maximum Common Subgraph Search: A Redundancy-Reduced Backtracking Approach},
  author={Yu, Kaiqiang and Wang, Kaixin and Long, Cheng and Lakshmanan, Laks and Cheng, Reynold},
  journal={Proceedings of the ACM on Management of Data},
  volume={3},
  number={3},
  pages={1--27},
  year={2025},
  publisher={ACM New York, NY, USA}
}

@InProceedings{dsb2025,
  author =	{Kothalawala, Buddhi W. and Koehler, Henning and Wang, Qing},
  title =	{{Learning to Bound for Maximum Common Subgraph Algorithms}},
  booktitle =	{31st International Conference on Principles and Practice of Constraint Programming (CP 2025)},
  pages =	{22:1--22:18},
  series =	{Leibniz International Proceedings in Informatics (LIPIcs)},
  ISBN =	{978-3-95977-380-5},
  ISSN =	{1868-8969},
  year =	{2025},
  volume =	{340},
  editor =	{de la Banda, Maria Garcia},
  publisher =	{Schloss Dagstuhl -- Leibniz-Zentrum f{\"u}r Informatik},
  address =	{Dagstuhl, Germany},
  URL =		{https://drops.dagstuhl.de/entities/document/10.4230/LIPIcs.CP.2025.22},
  URN =		{urn:nbn:de:0030-drops-238837},
  doi =		{10.4230/LIPIcs.CP.2025.22}
}

@book{boostlib2011,
author = {Schling, Boris},
title = {The Boost C++ Libraries},
year = {2011},
isbn = {0982219199},
publisher = {XML Press}
}

@inproceedings{petsinis2025alpine,
author = {Petsinis, Petros and Skitsas, Konstantinos and Ranu, Sayan and Mottin, Davide and Karras, Panagiotis},
title = {Alpine: Partial Unlabeled Graph Alignment},
year = {2025},
isbn = {9798400714542},
publisher = {Association for Computing Machinery},
address = {New York, NY, USA},
url = {https://doi.org/10.1145/3711896.3736839},
doi = {10.1145/3711896.3736839},
booktitle = {Proceedings of the 31st ACM SIGKDD Conference on Knowledge Discovery and Data Mining V.2},
pages = {2315–2325},
numpages = {11},
location = {Toronto ON, Canada},
series = {KDD '25}
}

@inbook{zampelli2007symmetry,
author = {Zampelli, Stéphane and Deville, Yves and Dupont, Pierre},
publisher = {John Wiley \& Sons, Ltd},
isbn = {9780470612309},
title = {Symmetry Breaking in Subgraph Pattern Matching},
booktitle = {Trends in Constraint Programming},
chapter = {10},
pages = {203-218},
doi = {https://doi.org/10.1002/9780470612309.ch10},
url = {https://onlinelibrary.wiley.com/doi/abs/10.1002/9780470612309.ch10},
eprint = {https://onlinelibrary.wiley.com/doi/pdf/10.1002/9780470612309.ch10},
year = {2007}
}

\end{document}